\documentclass[11pt]{article}
\usepackage{times}
\usepackage{amsmath,amsthm, amssymb}
\usepackage[ruled,vlined]{algorithm2e}
\usepackage{fullpage}
\usepackage{color}
\usepackage{tikz}
\usetikzlibrary{arrows.meta,positioning}
\usepackage{algorithmic}
\usepackage{setspace}
\usepackage{enumitem}
\usepackage{wrapfig}
\usepackage{subcaption}
\usepackage{hyperref}
\usepackage[numbers]{natbib}

\usepackage{comment} 

\newcommand{\scg}[1]{\texttt{SCG}}

\newtheorem{theorem}{Theorem} 

\newtheorem{lemma}{Lemma}
\newtheorem{proposition}{Proposition}

\newtheorem{definition}{Definition}
\newtheorem{remark}{Remark}

\newtheorem{conjecture}{Conjecture}
\newtheorem*{conjecture*}{Conjecture}
\newtheoremstyle{nonindented}{1ex}{1ex}{}{}{\bfseries}{.}{.5em}{}
\newtheoremstyle{indented}{1ex}{1ex}{\itshape\addtolength{\leftskip}{0.6cm}\addtolength{\rightskip}{0.6cm}}{}{\bfseries}{.}{.5em}{}
\theoremstyle{nonindented}
\theoremstyle{indented}

\theoremstyle{plain}

\newcommand{\bv}{\mathbf}

\newcommand{\set}[1]{\left\{ #1 \right\}}

\renewcommand{\hat}{\widehat}
\renewcommand{\tilde}{\widetilde}
\renewcommand{\bar}{\overline}

\newcommand{\bvec}[1]{\boldsymbol{ #1 }}

\DeclareMathOperator{\poly}{poly}

\def\min{\qopname\relax n{min}}

\def\Pr{\qopname\relax n{\mathbf{Pr}}}

\newcommand{\RR}{\mathbb{R}}

\newcommand{\PP}{\mathbb{P}}

\def\C{\mathcal{C}}

\def\O{\mathcal{O}}

\def\eps{\epsilon}

\newcommand{\INPUT}{\item[\textbf{Input:}]}
\newcommand{\OUTPUT}{\item[\textbf{Output:}]}

\newcommand{\mini}[1]{\mbox{minimize} & {#1} &\\}
\newcommand{\maxi}[1]{\mbox{maximize} & {#1 } & \\}

\newcommand{\st}{\mbox{subject to} }
\newcommand{\con}[1]{&#1 & \\}
\newcommand{\qcon}[2]{&#1, & \mbox{for } #2.  \\}
\newenvironment{lp}{\begin{equation}  \begin{array}{lll}}{\end{array}\end{equation}}
\newenvironment{lp*}{\begin{equation*}  \begin{array}{lll}}{\end{array}\end{equation*}}

\title{Algorithmic Information Design in Multi-Player Games: \\ Possibility and Limits in Singleton Congestion} 
\author { 
Chenghan Zhou %\\ \rule{-0.2in}{0pt} {\small Department of Computer Science} 
    \\  Princeton University  \\  chenghanzh@princeton.edu
    \and Thanh H. Nguyen %\\ {\small Department of Computer Science} 
    \\ {  University of Oregon} \\ thanhhng@cs.uoregon.edu
    \and Haifeng Xu %\\ {\small Department of Computer Science} 
    \\ University of Chicago  \\ haifengxu@uchicago.edu
}
\date{}

\begin{document}
\begin{titlepage}
\clearpage\maketitle
\thispagestyle{empty}
\maketitle

\begin{abstract}
 
 Most algorithmic studies on multi-agent information design so far have focused on the restricted situation  with no inter-agent externalities; a few exceptions investigated truly strategic games such as zero-sum games and second-price auctions but  have all focused only on optimal  public  signaling.  This paper initiates the  algorithmic information design of both \emph{public} and \emph{private} signaling in a fundamental class of games with negative externalities, i.e.,   singleton congestion games, with wide application in today's digital economy, machine scheduling, routing, etc.   
 
 % Besides its fundamentalness and wide applications,  this class of games  is  ideal  for the study of information design since its Nash equilibrium (even the socially-optimal one) can be computed efficiently, which thus disentangles the equilibrium complexity from the complexity of information design. 
 
 For both public and private signaling, we show that the optimal information design can be efficiently computed when the number of resources is a constant. To our knowledge, this is the first set of    efficient \emph{exact} algorithms for information design in succinctly representable many-player games. Our results hinge on novel techniques such as developing certain ``reduced forms'' to compactly characterize equilibria in public signaling  or to represent  players' marginal beliefs in private signaling. When there are many resources, we show computational intractability results. To overcome the issue of multiple equilibria, here we introduce a new notion of  equilibrium-\emph{oblivious} hardness, which rules out any possibility of computing a good signaling scheme, irrespective of the equilibrium selection rule.  

\end{abstract}

\end{titlepage}

\section{Introduction}

 In today's digital economy, there are numerous situations where many players have to compete for limited resources. For instance, on ride-hailing platforms  such as Uber and Lyft, drivers pick an area to go and then compete with other drivers for riding requests at that area; on content platforms such as Youtube and Tiktok, content providers choose a style/theme for their contents and then compete with other providers of the same theme   for Internet traffic interested in that theme; on digital markets such as Amazon and Wayfair, retailers choose a particular product category (e.g., pet supplies or home\&kitchen, etc.) to focus on and compete with other retailers   for sale demands on that category. All these problems  share the following   similarity: (1) many players make a choice (e.g., a  ride-sharing area or a content theme) from multiple options and their payoffs has negative externalities with other players of the same choice due to competition; (2) players have high uncertainty about the payoffs of their choices since the entire system's  demand of riding requests or Internet traffic are unknown to an individual player, whereas the system usually has much fined-grained information about these uncertainties. An important operational task common in all these applications is the following: how can the system 
 %Many real-world strategic settings are rifle with uncertainty and information asymmetry --- i.e.,  some agents have privileged access to certain private information whereas others do not. This raises the fundamental   question of how an agent 
 (the \emph{sender}) strategically reveal her privileged information to influence the decisions of so many  players  (the \emph{receivers}) in order to steer their collective decisions towards a desirable social outcome?  This task, also known as   \emph{information design} or \emph{persuasion}   \cite{kamenica2011bayesian,Bergemann16Bayes,Dughmi2017,candogan2020information}, has attracted extensive recent interests. Besides the aforementioned problems, it has found application in many other domains including auctions \cite{ Emek12,li2019signal,Badanidiyuru2018targeting},  recommender systems \cite{ mansour2020bayesian,Mansour2016bayesian,yan2020warn}, robot planning   \cite{ KerenXKPG20},    traffic congestion control \cite{ Bhaskar2016,das2017reducing},  security \cite{ Xu15,Rabinovich15,Xu18},   and recently  reinforcement learning \cite{ simchowitz2021exploration}.  
 
 Similar  to  mechanism design,  information design is also an optimization question  subject to incentive constraints. Thus unsurprisingly, it has  attracted much algorithmic studies, particularly in the  challenging situation of multi-receiver persuasion \cite{Bergemann2019}. Much algorithmic investigation has been devoted to the special case with \emph{no inter-agent externalities} \cite{ babichenko2016,Dughmi2017algorithmic,Xu20,celli20,castiglioni2020public}, i.e., a receiver's utility is not affected by other receivers' actions. This restriction is certainly not ideal, but does come with a reason. Indeed,    even in such case with no externalities, the optimal information design is already notoriously intractable. Specifically, it was shown to be NP-hard to obtain any constant approximation  if the sender  sends   a public signal to all receivers, a.k.a., \emph{public} signaling \cite{ Dughmi2017algorithmic,Xu20}. This hardness holds even when each receiver only has a binary action from $\{0, 1\}$ and when the sender simply wants to maximize the number of receivers taking action $1$ \cite{ Dughmi2017algorithmic}. While  public signaling may sometimes be desirable due to concerns of unfairness and discrimination or due to communication restrictions \cite{ dughmi2014hardness,xu2020tractability},  there are also   situations  in which  the sender  may send different signals  to different receivers separably, i.e., \emph{private} signaling.  This turns out to be more tractable:  optimal private signaling admits polynomial time algorithm so long as the sender's objective function can be efficiently maximized \cite{ Dughmi2017algorithmic,celli20} but becomes NP-hard otherwise \cite{ babichenko2016}. 
 
 Given the aforementioned hardness   for the no externality situation, it is less surprising that   optimal information design has received much less attention in strategic games \emph{with externalities}. Studies of information design in  games have so far mostly focused on \emph{public} signaling in restricted classes of games such as zero-sum games studied by \citet{ dughmi2014hardness}, non-atomic congestion games with linear latency functions studied by \citet{ bhaskar2016hardness} and second-price auctions studied by \citet{ Emek12}. Unfortunately, these works all exhibit sweeping intractability results. An interesting exception is the very recent work by \citet{ griesbach2022public}, which develops polynomial-time optimal public signaling schemes for multi-commodity non-atomic congestion games in the situation with parallel links, constant number of states and affine latency functions.  On the other hand, the study of \emph{private} signaling in general games have received significantly less attention. As part of their learning algorithm,   \citet{ Mansour2016bayesian} give a linear program (LP) for computing the optimal private scheme. However, the size of their LP is \emph{exponential} in the number of receivers. Overcoming this exponential dependence on the number of agents is an intrinsic challenge in the design of optimal private signaling scheme (see more discussions below).

 This paper initiates a systematic algorithmic investigation of both public and private signaling in \emph{succinctly representable multiplayer games} and focus on a basic  class of strategic games, i.e., the  atomic   singleton congestion games  (\scg{}s).  We adopt the perspective of a social planner who looks to use information design to minimize the total social cost, a widely studied global objective in congestion games.  Congestion games succinctly capture \emph{negative} externalities among agents.  The \scg{} is an important   special case of    congestion games where each player's action is a singleton set of the resources, i.e., a single resource. While the game class of \scg{} may appear ``narrow'' at the first glance, it is a very fundamental and widely-studied class of games  --- it has been the sole subject in many previous papers, including the arguably influential work by \citet{koutsoupias1999worst} which introduced the concept of the price of anarchy as well as its notable follow-up work by \citet{czumaj2007tight}.  Therefore, we believe a thorough study of this basic class of games is an important step towards understanding the algorithmics of optimal information design in truly strategic setups. Moreover, in addition to the wide application of \scg{} mentioned at the beginning of this section, it also finds   application in other domains such as  traffic routing \cite{ griesbach2022public}, job scheduling \cite{ ieong2005fast, Gairing2006ThePOA}, firm competition   \cite{ gairing2007latency} and  communication over networks \cite{ aland2006exact,ackermann2006PureNash}. 
 
 % In \scg{}, agents' action set (i.e., a subset of resources) is allowed to be different in general; the game is called \emph{symmetric} \scg{} if all agents' action sets are the same, which without loss of generality can be the entire set of available resources. 
 
 Besides its wide applicability, there are also multiple  more basic reasons that \scg{}s are an ideal game class for the study of optimal information design for succinctly represented multiplayer games.  First, it is a fundamental class of games, which as we show already exhibits quite non-trivial computational challenges. Therefore, a thorough algorithmic understanding for this elemental  class is essential for the study of information design   in more complex setups. Second,  a celebrated work by \citet{ ieong2005fast} shows that   various types of Nash equilibria (NEs), including both the socially-optimal NE and the potential-function-minimizing NE, can be computed efficiently in \scg{}s with \emph{arbitrary} cost functions, whereas \citet{ fabrikant2004complexity} prove that computing a pure Nash equilibrium becomes PLS-complete in general congestion games.  Therefore, the restriction to \scg{}s   allows us to ``disentangle'' the complexity study of information design from the complexity of computing the equilibrium.  Third, information design in   such class of strategic games with many players gives rise to several fundamental new challenges that has  not been present in previous work and thus requires the introduction of novel concepts and techniques which we now elaborate.
 
 Finally, recent study by \citet{ nachbar2020power} shows that the welfare improvement via information design can never exceed the \emph{price of anarchy} of the underlying base games.  For example, in non-atomic routing with linear latency, the maximum possible social cost \emph{reduction} by using any signaling scheme is at most $1/3=4/3-1$ fraction of the social  optimum, since the price of anarchy of the games is at most $4/3$ \cite{ roughgarden2015intrinsic}.  Therefore, information design would be  more useful in games with large price of anarchy, which is true for the general \scg{} games we consider. The following simple example   illustrates public and private signaling in \scg{}s and  how it may significantly reduce social cost. 
 
 \subsection{An illustrative example. } Consider the example in   Figure \ref{fig:signal-example}. The \scg{} has   $N=2$ agents, $3$ resources and two   state $\theta_1,\theta_2$ of equal probability $0.5$.  The $2$ agents both start from source $s$ and each picks one edge to sink $t$. Cost function $c(0)=c(1) = 0$ whereas $c(2) = 1$; $\epsilon$ is an arbitrarily small positive number.   
 \begin{figure}[h]
 	\vspace{-4mm}
 	\begin{center}
 		\includegraphics[width=0.6\textwidth]{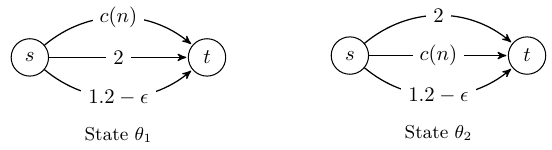}
 		\caption{An   \scg{} with $N=2$ symmetric agents and $R=3$ resources; $c(0) = c(1) =0$ and $c(2) = 1$. }
 		\label{fig:signal-example}
 	\end{center}
 \end{figure}
 
 On one hand, under the transparent policy of \emph{full  information}, both agents can distinguish the state $\theta_1, \theta_2$ and thus will both choose the edge with cost $c(n)$, leading to congestion on the $c(n)$ edge and total social cost $2\times c(2) = 2$. On the other hand, the policy of \emph{no information} will lead to the two agents take the top and middle edge respectively at equilibrium, leading to total social cost $2$ as well.   Simple exercise shows that the optimal \emph{public} signaling scheme --- i.e., when all agents receive the same information--- will mix $0.8$ fraction of $\theta_1$ and $0.2$ fraction of $\theta_2$ for one signal $\sigma_1$, and $0.2$ fraction of $\theta_1$ and $0.8$ fraction of $\theta_2$ for another signal$\sigma_2$. Conditioned on  public signal $\sigma_1$ (similar analysis applies to $\sigma_2$), at equilibrium one agent will take the top path and another take the bottom path, leading to expected total cost $1.2 - \epsilon + [0.8\times c(1) + 0.2 \times 2] = 1.6-\epsilon$.  
 
 Interestingly, a simple \emph{private} signaling scheme, by revealing  no information to one agent (may be picked randomly) but revealing full information to another, can achieve the minimum social cost $1.2 - \epsilon +0 = 1.2- \epsilon$.  This is because   the agent who receives full information will be able to precisely identify the edge with cost $c(n)$ and then take it. This leaves the bottom ($1.2- \epsilon$)-cost edge as a best  resource for the other agent with no information. Optimal private signaling reduces the equilibrium social cost under full information by $2 - 1.2 + \epsilon = 0.8 + \epsilon$, or equivalently by $\frac{0.8 + \epsilon}{1.2 -   \epsilon} \approx 67\%$ fraction of the social optimum. This ratio can be significantly increased by considering the same network above but with latency functions that lead to higher price of anarchy   (e.g., $c(n) = (n/N)^d$ with large $N$), though   the algebraic calculation there will be more involved and less intuitive. The above example also illustrates how much more powerful private signaling may have   than the   more-often-studied public signaling in strategic games  as in the previous literature.

 \subsection{Overview of results, challenges, and our techniques.} 
 
 We adopt the perspective of an informationally advantaged social planner (the sender) looking to design signaling scheme to minimize the social cost in \scg{}s. Our main results are a complete characterization about the algorithmics of the sender's optimization problem. Specifically, when there are constant number of resources, we show that both the optimal public and private signaling can be computed in polynomial time. 
 Notably,  in hindsight, the case of small number of resources does not at all imply that the information design problem may be easy.  As mentioned previously, even in the case with binary receiver actions and no agent externalities, the space is already rifle with hardness results for both public and private signaling as shown by \citet{ babichenko2016,Dughmi2017algorithmic}. We will elaborate
 next why agent externality makes optimal signaling even harder.  Indeed, recent work by \citet{ yang2019information}  studies similar resource competition problems in  singleton congestion games like us. However, their polynomial-time algorithms for public and private information  are devised under very restrictive assumptions that there are only two resources and moreover, one of the resources has to be a trivial one   with constant utility $0$.  Our results strictly generalize the algorithms by \citet{ yang2019information}.  
 %To our knowledge, this is the first time that an exactly optimal public and private signaling scheme can be efficiently computed in a succinctly represented large games with many agents (we will elaborate next why agent externality makes optimal   signaling significantly harder). 

 When there are many resources, our results are negative; we show that even in symmetric \scg{}s, (1) it is NP-hard to design a fully polynomial time approximation scheme (FPTAS) for   optimal public signaling   in a very strong sense which we detail next; (2) the separation oracle for the dual problem of the optimal private signaling scheme is NP-hard even in symmetric \scg{}s with affine cost functions.  
 
 {\bf An intriguing   conjecture.} An open problem left from our results is the complexity of the original problem of optimal private signaling. When there is only a single state of nature, this problem degenerates to computing the optimal correlated equilibrium of a singleton congestion game, which  surprisingly is still  an open question to date as well. This open question is particularly intriguing for the special case of \emph{symmetric} \scg{}s with affine cost functions,   because both the socially-optimal coarse correlated equilibrium (CCE) and the socially-optimal Nash equilibrium of an \scg{}  were shown to admit a polynomial time algorithm by \citet{Castiglioni2020Signaling} and  \citet{ieong2005fast}, respectively. However, our negative result for the separation oracle of the corresponding optimal correlated equilibrium seems to suggest the conjecture of the hardness of optimal CE. If this conjecture was true, such situation of intractable CE yet tractable NE and CCE is a very rare phenomenon --- to the best of our knowledge, there is no   class of games with such complexity property that are known so far.  %  Our result shows that the separation oracle for the dual problem of  optimal correlated equilibrium is NP-hard, which serves as  evidence for our conjecture that the optimal correlated equilibrium (CE) may also be intractable.      
 
 Next we elaborate on the key challenges in proving the above results and our techniques to address them.  The first intrinsic challenge is the issue of \emph{equilibrium selection} during the design of  the optimal public   signaling  scheme. The \scg{} is known to admit multiple  Nash equilibria. Then a key question   is, given any public  signal which induces an expected \scg{} games among agents, which NE we should posit  the agents to play. Previous studies of  public signaling all bypassed this issue by either assuming    no agent externality or considering the situation with a unique Nash equilibrium (e.g., zero-sum games \cite{ dughmi2014hardness} and non-atomic routing \cite{ Bhaskar2016}). Consequently their techniques cannot be easily adapted to many other situations in which equilibria are not unique or do not admit an obvious selection rule.   In this paper, we for the first time directly tackle  this issues  and develop   algorithmics based on the nature of the results, as follows:
 \begin{itemize}
 	\item On one hand, for \emph{positive} result of efficient algorithms, it is necessary to adopt certain equilibrium selection rule since it is challenging (if not impossible) to design an efficient algorithm that works under arbitrary equilibrium selection  (\scg{} may have exponentially many equilibria \cite{ ieong2005fast}).  Therefore,  our algorithm  follows the convention of  the information design   literature \cite{ Bergemann16Bayes,Taneva2015,yang2019information}  and   adopts the optimistic equilibrium, i.e., the socially-optimal NE. % due to two reasons: (1) this NE can always be   identified without any ambiguity; (2) we believe the sender as the social planner may have the power to intervene the NE by suggesting the NE to all agents.\footnote{The conservative choice of the socially-\emph{worst} NE is an  interesting direction for future research, though we believe the optimistic equilibrium selection appears relatively more appealing. Similar optimistic equilibrium selection has also been used in relevant recent works with two resources \citet{ Yang19}. } We remark that adopting an arbitrary NE is not ideal since the   \emph{price of anarchy}  in \scg{}s can be very large \citet{ gairing2006price}.\footnote{This is   another interesting reason for   studying information design in \scg{}s.  Nachbar and Xu \citet{ nachbar2020power} prove that the effectiveness of information design is always upper bounded by the \emph{price of anarchy} in any class of games. Therefore,  information design can only be a very effective ``knob'' to improve sender objectives in games with large PoA.}  
 	\item On the other hand, in order to prove convincing \emph{negative} result of computational hardness,   equilibrium selection issue becomes much trickier to handle. Even we proved the hardness of optimal signaling under certain equilibrium selection, it does not imply  the hardness under a slightly altered equilibrium selection rule. %In fact, it may even not be clear that whether the hardness comes from finding that specific equilibrium or from designing the information structure.  
 	To overcome this challenge,  we introduce a novel notion of \emph{equilibrium-oblivious  inapproximability}. Intuitively, we say that optimal public signaling is \emph{equilibrium-obliviously} $\alpha$-inapproximable if there is no $\alpha$-approximate algorithm \emph{regardless of} which NE one adopts in \emph{any} signaling scheme. Such an inapproximability result  completely rules out any possibility of designing a good public signaling scheme, irrespective of the equilibrium selection.  We believe this novel concept of equilibrium-oblivious intractability may be of independent interest for future works to bypass the equilibrium-selection issues in hardness proofs, and our result illustrates the possibility of achieving this goal. 
 \end{itemize}

 The second key challenge is the issue of exponential dependence of the private signaling scheme on the number of agents. Notably, this is also the central challenge  in the computation of an \emph{optimal correlated equilibrium}, which is well-known to be notoriously challenging and is NP-hard in many classes of succinct games such as general congestion games, facility location games and scheduling games \cite{ papadimitriou2008computing}.   However,   optimal private signaling is arguably even more difficult since it  contains the optimal correlated  equilibrium as a strict special case, when 
 there is only a single state of  nature. 
 Like previous works by \citet{ Mansour2016bayesian,Dughmi2017algorithmic,celli20}, we adopt the solution concept of  Bayes correlated equilibrium due to \citet{ Bergemann2016information}, which characterizes signaling as obedient action recommendation for each receiver.  Due to the exponential blowup in the total number of possible action profiles, this gives rise to a linear optimization problem with exponentially many variables. The agent externality makes it crucial  to characterize each agent's  posterior beliefs about other agents since an agent's utility here depends on other agents' beliefs as well as their   actions. Notably, this complication is absent in previous private signaling setting with no inter-agent externality \cite{ Dughmi2017algorithmic,babichenko2016,celli20}, which makes the design problem there significantly easier.  To overcome this challenge, we employ the idea of ``reduced form''  from mechanism design \cite{ cai2012algorithmic,alaei2012bayesian,border1991implementation} to characterize each agent's marginal belief about other agents' actions. We characterize the feasibility constraints of the reduced forms.   %This characterization crucially hinges    on the structure of \scg{}s, in which the cost at a resource only depends on the number, not the identity, of the agents choosing it. 
 En route,  we also develop an efficient algorithm  to sample the optimal private signaling scheme on the fly, which strictly generalizes a classic sampling technique  by  \citet{ tille1996elimination}  in the statistics literature, and may be of independent interest. To our knowledge, this is the first time that reduced form   is used for information design with many interacting receivers.  In the main body, we illustrates the similarities and differences between the reduced form for information design and that for auction design. We hope this discussion could spur more applications of reduced form to information design.

 \subsection{Additional related work.}
 % Discuss relation with other related works on information design. This part can refer to \citet{ yang2019information} and \citet{ xu2020tractability}. Try to not copy any sentences, but instead just rephrase their sentences.  
 Information design in games has attracted much recent attention in the economics literature \cite{ Alonso14,Bergemann16Bayes,Bergemann2016information,Taneva2015,mathevet2020information}. Most of these works have focused on understanding the properties of the optimal signaling scheme. Specifically, our algorithms leverages the notion of Bayes correlated equilibrium by  \citet{ Bergemann16Bayes}, which characterizes the set of all possible Bayes Nash equilibrium under private signaling.  \citet{ mathevet2020information} highlight  the challenge of information design with non-trivial equilibrium selection rules and  characterize this task as a two-level optimization problem.     
 As mentioned previously, most algorithmic studies so far have bypassed the issue of equilibrium selection by focusing on either games with unique equilibrium or setting with no agent externalities. However, equilibrium selection cannot be bypassed in \scg{}s (which is also a key reason that the price of anarchy and stability is studied extensively in congestion games \cite{ roughgarden2004bad,roughgarden2005selfish,christodoulou2005anarchy}), and thus our work directly tackle this issue   in information design. %Algorithmic information design has attracted much recent attention. We refer the reader to \citet{ dughmi2017survey} for a survey and to \citet{ candogan2020information} for an overview of applications in operations. However, there has been very limited understanding on algorithmic information design  in congestion games, except the hardness result of Bhaskar et al. \citet{ bhaskar2016hardness} for public signaling in non-atomic congestion games with linear latency functions.  Notably,  recent work by Castiglioni et al. \citet{ Castiglioni2020Signaling} studied private signaling in singleton congestion games but with significantly \emph{relaxed} player incentives, under the notion of \emph{ax-ante} private signaling, and thus is   not comparable with our work.   
 %  Bayesian Persuasion   provides a principled framework to study how the disclosure of information can change agents' beliefs and behavior. The original model proposed by Kamenica and Gentzkow~\citet{ kamenica2011bayesian} comprises a sender trying to influence a single receiver's behavior by selectively revealing his information. The sender commits to a signaling scheme before accessing  the information. %This assumption is realistic in ubiquitous settings. e.g., when the platform cares about long-term reputation, or when regulations enforce the sender to commit. 
 % Subsequent works generalize this model to one-sender-many-receivers  setting.  While there is no general efficient algorithm to compute a social welfare maximizing signaling scheme yet, existing work studies computational tractability of specific problems under the framework of Bayesian Persuasion~\citet{ Dughmi2017algorithmic,Castiglioni2020Signaling}.
 
 Congestion games are a fundamental class of succinctly represented multiplayer games and have been studied extensively. % Singleton congestion games is a  subclass of congestion games, which are characterized as a special type of potential games and thus always admit a pure strategy Nash equilibrium~\citet{ rosenthal}. 
 Much previous algorithmic research  has focused on the complexity   of equilibrium in congestion games \cite{ fabrikant2004complexity,meyers2012complexity}.  \citet{ ieong2005fast} introduced the class of singleton congestion games and designed an efficient dynamic programming algorithm to compute the socially-optimal Nash Equilibrium. Closely related to ours is the recent work  by \citet{ yang2019information}; motivated by spatial resource competition, they study public and private signaling in a special case of ours, i.e., two resources and one resource has constant utility, and developed polynomial time algorithms for both optimal public and private signaling. The present work strictly generalizes their algorithmic results. Concurrent work by \citet{griesbach2022public} developed polynomial time algorithms for optimal public signaling in non-atomic routing games with parallel links and linear latency functions. The singleton congestion game  can be viewed as atomic routing with parallel links, but we study both  public and private signaling and allow arbitrary latency functions. Recent work by  \citet{ Castiglioni2020Signaling} studied private signaling in singleton congestion games but with significantly \emph{relaxed} player incentives, under the notion of \emph{ex-ante} private signaling, and thus is   not comparable with our work.     \citet{ marchesi2019leadership} study the situation when one of the agent can commit in singleton congestion games. Congestion games have also been extensively studied in  the price of anarchy/stability literature \cite{ roughgarden2004bad,roughgarden2005selfish,christodoulou2005anarchy}.

 Finally, on the technical side, our algorithm for optimal private signaling relies on the study of the ``reduced forms''.  Recent work by   \citet{ candogan2020reduced} introduced the idea of reduced form for information design but studied a completely different setup with a single receiver and continuous action space. The reduced form there is employed to compactly ``summarize'' the mean of any posterior distribution supported on an continuum space. However, the reduced form in our setting is used to compactly describes one agent's belief about other agents' report, which is more analogous to its use in classic auction design \cite{ cai2012algorithmic,alaei2012bayesian,border1991implementation}.

 \section{Preliminaries.} 
 % \todo{State the general model here, define public signaling, private signaling and ex-ante private signaling}
 
 % \hf{Since we are not able to figure out the spatial resource competition, so let's switch the entire paper to Singleton Congestion, so no Spatial Resource Competition (\texttt{SRC}) at all hereafter. I changed the title now to reflect that we intent to provide a thogrough study about the fundamental model of singleton congestion games.  

 	% Please use "resource" instead of "location", use "congestion" instead of "award" or "payoff". Please refer to the paper "Fast and Compact: A Simple Class Of Congestion Games" for all terminologies. Also we will not consider costs, as opposed to rewards, could you please change through the entire paper about this notation? Specifically, for our NP-hard reductions, we will need to change all functions to be the \textbf{negation} of the function (add a constant if they are negative) in order to reflect that they are costs now. }

 \subsection{ Basics of singleton congestion games (\scg{}s).} \label{sec:model:atomic}
 
 A   singleton congestion game (\texttt{SCG}) consists $N$ \emph{agents} denoted by set $[N] = \{1, \cdots, N \}$ and $R$ \emph{resources} denoted by set $[R]  = \{1, \cdots, R\}$.  Each resource $r$ is associated with a \emph{non-negative}, monotone \emph{non-decreasing} congestion function $c_r: [N] \to \RR^+$.  Each agent $i$ has a set of available actions $A_i \subseteq [R]$ and is allowed to  choose a single resource  from $ A_i$ (thus the  ``singleton'' in its name). Any agent at resource $r$ suffers cost $c_r(n_r)$  where $n_r$ is the total number of agents picking resource $r$.   % Upon choosing $r \in A_i$, the agent induces a cost that depends on the total number of agents choosing $r$.
 Each agent $i \in [N]$ simultaneously chooses an action $a_i \in A_i$.  By convention, we use $\bv{a} = (a_1, \cdots, a_N)$, sometimes denoted as $(a_i, a_{-i})$ to emphasize agent $i$, to denote the profile of actions chosen by all agents. Let $A = A_1 \times \cdots \times A_N$ denote the set of all possible action profiles. Note that $A$ may have size  $\Omega(R^N)$ due to the combinatorial explosion of different agents' choices.  
 
 Let $\bv{n} = (n_1, \cdots, n_R)$ denote the profile of the numbers of agents at all resources, and is referred to as a \emph{configuration}. Any action profile $\bv{a}$ uniquely determines a configuration $\bv{n}$, consisting of the numbers of agents choosing each resource. Specifically, we denote $n_r(\bv{a}) = |\{i: a_i = r\}|$, or simply write $n_r$ or $\bv{n}$ when $\bv{a}$ is clear from the context.  A trivial constraint for any non-negative $\bv{n}$ to be a \emph{feasible configuration} is that $\sum_{r=1}^R n_r = N$. However, not all such  $\bv{n}$'s are feasible since each agent $i$'s action space $A_i$ is restricted.    We use $P(A)$ to denote the set of all feasible configurations $\bv{n} \in P(A)$ given the action space of agents. $P(A)$ is of size at most $\O(N^R)$ since there are at most $\O(N^R)$ ways to partition $N$ into $R$ non-negative integers. When $R$ is a small constant, the size of $P(A)$ may be significantly smaller than that of $A$.   
 
 Given any action profile $\bv{a}$, the cost of agent $i$ is the congestion  associated with resource $a_i$. All agents attempt to minimize their own congestion cost. 
 Thus, we write the utility of agent $i$ to be $u_i(\bv{a}) = -c_{a_i}(n_{a_i})$. An action profile $\bv{a}^*$ is called a pure Nash equilibrium if $a_i^*$ is a best response to $a_{-i}^*$ for each agent, or formally, $u_i(a_i^*, a_{-i}^*) \geq  u_i(a_i, a_{-i}^*)$ for any $i \in [N]$ and $a_i \in A_i$. One may also consider mixed strategies. However, \citet{ rosenthal} shows that a key property of congestion games is that they always admit at least one pure strategy Nash equilibrium. In particular, the action profile minimizing the  \emph{potential function} $\Phi(\bv{a}) = \sum_{r \in [R]} \sum_{j=1}^{n_{r}(\bv{a})} c_{r}(j)$ always forms a pure Nash equilibrium. Therefore, throughout the paper we should always consider pure strategy NE, which is more natural to adopt whenever it exists \cite{ rosenthal}.  
 
 The total \emph{social cost} induced by strategy profile $\bv{a}$, regardless it is a NE or not, is defined to be the sum of all agents' congestion, i.e., 
 \begin{equation}\label{eq:sw-def}
 	\texttt{SC}(\bv{a}) = \sum_{i \in [N]}  c_{a_i}(n_{a_i}(\bv{a})) = \sum_{r \in [R]}  n_{r}(\bv{a}) c_{r}(n_{r}(\bv{a}))
 \end{equation}
 
 One special class of   \scg{}s  is the \emph{symmetric} singleton congestion games, also known as  congestion game with parallel links  (see Figure \ref{fig:signal-example} for an example). That is, the   action sets $A_i$ are the same across agents. In this case, w.l.o.g., we will let $A_i = R, \forall i$ and thus $A = [R]^N$ is the set of all action profiles.  % We will see later that how this restriction may restore computational tractability for some situations when the optimal signaling is NP-hard in general. 

 \subsection{Uncertainties and information design.}
 \label{sec:model:schemes}
 This paper concerns singleton congestion games with uncertainty. Specifically, the congestion function of any resource $r$ depends on a common \emph{random} state of nature $\theta$ drawn from support $\Theta$ with prior distribution $\mu \in \Delta_{\Theta} =\{ \bv{p} \in \RR^{\Theta}_+: \sum_{\theta \in \Theta} p_{\theta} = 1 \} $.  Let $\mu_\theta$ denote the probability of state $\theta$ and  $c^{\theta}_r$ denote the congestion function of resource $r$ at state $\theta$. We adopt the perspective of an information-advantaged social planner, referred to as the \emph{principal}, who has privileged access to the \emph{realized} state $\theta$. After observing $\theta$, the principal would like to strategically reveal this information to agents   in order to influence their actions. The principal is equipped with the natural objective of minimizing the equilibrium social cost, i.e., the sum of the players' costs at equilibrium. 
 
 \subsubsection*{Public signaling Schemes, and the perspective of prior decomposition.}
 We adopt the standard   assumption of information design \cite{ bergemann2016bayes}, and assume that the principal commits to a signaling scheme $\pi$ before $\theta$ is realized and $\pi$ is publicly known to all agents. At a high level, a public signaling scheme generates another random variable, called the \emph{signal} $\sigma$, that is correlated with $\theta$. Therefore, as the signal $\sigma \in \Sigma$ is realized and sent  through a public channel, all agents will infer the same partial information about 
 $\theta$ due to its correlation with $\sigma$.  Formally,  a public signaling scheme can be described by variables $\{\pi(\sigma|\theta)\}_{\sigma \in \Sigma, \theta \in \Theta}$, where $\pi(\sigma|\theta)$ is the probability of sending signal $\sigma$ conditioned on observing the state of nature $\theta$. The probability of sending signal $\sigma$ equals $\bv{Pr}(\sigma) = \sum_{\theta} \mu_{\theta} \pi(\sigma|\theta)$. Upon receiving signal $\sigma$, all agents perform a standard Bayesian update and infer the posterior probability about the state of nature $\theta$ as follows:  $$\bv{Pr}(\theta|\sigma) = \mu_{\theta}\pi(\sigma|\theta)/\bv{Pr}(\sigma).$$ 
 Therefore, each signal $\sigma$ is mathematically equivalent to a posterior distribution  $\bv{Pr}(\cdot|\sigma) \in \Delta_{\Theta}$. Moreover,  $ \sum_{\sigma}\Pr(\sigma) \bv{Pr}(\theta|\sigma) = \mu_{\theta}$. Therefore, the signaling scheme can be viewed as a convex decomposition of the prior  $\mu$ into a distribution over posteriors $\bv{Pr}(\cdot|\sigma) \in \Delta_{\Theta}$ with convex coefficient $\bv{Pr}(\sigma)$. As established by \citet{ aumann1995repeated,blackwell1953equivalent},  any such convex decomposition can be implemented as a signaling scheme as well, establishing their equivalence.    
 
 Since all players receive the same information, they will then play a singleton congestion game based on the expected congestion functions  $c_i^{\sigma}(\bv{a})= \sum_{\theta} \bv{Pr}(\theta|\sigma) c_i^\theta(\bv{a})$. Like  standard game-theoretic analysis in this space, agents are assumed to play a pure NE. 
 
 \subsubsection*{Private signaling schemes.}
 
 Private signaling relaxes public signaling by allowing the sender to send different, and possibly correlated, signals to different players. Specifically, let $\Sigma_i$ denote the set of possible signals to player $i$ and $\Sigma = \Sigma_1 \times \Sigma_2 \times \cdots \times \Sigma_N$ denote the set of all possible signaling profiles. With slightly abuse of notation, a private signaling scheme can be similarly captured by variables $\{\pi(\sigma|\theta)\}_{\theta \in \Theta. \sigma \in \Sigma}$. When signaling profile $\sigma$ is restricted to have the same signal to all agents, this degenerates to public signaling. 
 
 Private signaling leads to a  Bayesian game where each agent holds different information about the state of nature. Specifically, given a publicly known signaling scheme $\pi$, each agent $i$ infers a posterior belief over $\theta$ and $\sigma_{-i}$ after receiving $\sigma_i$:
 $$\bv{Pr}(\theta, \sigma_{-i} | \sigma_i) = \frac{\mu_{\theta}\pi(\sigma_i, \sigma_{-i}|\theta)}{\sum_{\theta' \in \Theta, \sigma_{-i}' \in \Sigma_i} \mu(\theta') \pi(\sigma_i, \sigma_{-i}' | \theta')}, \forall \theta \in \Theta, \sigma_{-i} \in \Sigma_{-i}$$
 Any private signaling scheme induces a Bayesian game among receivers, with beliefs as derived above. The agents then play a Bayesian Nash equilibrium. In both public and private signaling,  multiple  equilibria may exist; we will discuss about which NE to choose in corresponding sections.

 \section{The blessing of small number of resources.} \label{sec:constantR}
 
 In this section, we show that when the number of resources $R$ is a constant (but the number of states $|\Theta|$ and number of agents $N$ can be large), both optimal private and public signaling can be computed in polynomial time.  Despite the restriction to a small number of resources, we believe this is quite encouraging message for information design in succinctly representable large games.  % Indeed, both results hinges on non-trivial algorithmic techniques such as using ``reduced form'' to compactly characterize a signaling scheme.   
 In fact, previous studies about multi-receiver public and private signaling are both rife with hardness results even when there are only two actions for each receiver and even in the absence of  receiver externalities  \cite{ babichenko2016,Dughmi2017algorithmic,Xu20}.  This is because even when $R$ is a constant,   there are still $R^N$ possible action profiles and the asymmetry of agent action sets makes it important to pin down which agent picks which resource.  % Our algorithm   relies crucially on the special structure of \scg{}s and new ideas of using ``reduced form'' to compactly characterize a signaling scheme.     % This allows us to compute a ``reduced'' set of signals in both public and private signaling settings. As shown in later sections, the restriction on small number of resources is necessary, since computing the optimal signaling is NP-hard when the number of resources is large.
 
 \subsection{Optimal public signaling.}
 
 As mentioned previously, one key challenge of studying public signaling scheme is the existence of  multiple Nash equilibria:  given a posterior distribution, which equilibrium should we adopt?  %Note that applying the Bayesian persuasion framework does not eliminate this issue, because the principal is restricted to further affect agents' actions after manipulating their posterior belief. 
 Following the convention of information design \cite{ Bergemann16Bayes,Taneva2015,yang2019information}, our algorithm   adopts the optimistic Nash equilibrium, i.e., the Nash equilibrium that minimize the social cost.  That is,  the principal as a social planner is assumed to have the power to influence agents' behaviors by ``recommending''  an equilibrium under any  public signal \cite{ Bergemann16Bayes}. %\footnote{We remark that whichever equilibrium selection rule one adopts, it has to be uniquely identifiable under \emph{any public signal}.}  %To satisfy this criteria,  other alternatives one could possibly choose are the one maximizing the potential function or the most pessimistic equilibrium with largest social cost. While both are less natural (and also less considered) than the optimistic equilibrium selection,  we nevertheless show in  Appendix \ref{append:constant-public-potential} that the optimal public signaling scheme under potential-maximizing equilibrium can also be computed in polynomial time, whereas leave public signaling under pessimistic equilibrium selection as an interesting open question.
 Our main result here is an efficient algorithm for optimal public signaling  when $R$ is a constant.  
 
 % We will show that the optimal public signaling scheme degenerates to recommendation of configurations under the optimistic assumption. In addition, since the potential function minimizing Nash equilibrium characterize a special property in \texttt{SCG}s, we also provide a method to compute the optimal public signaling scheme when the equilibrium is chosen as the potential function minimizer, but defer details to the appendix \ref{append:constant-public-potential}.

 % \begin{theorem}\label{thm:opt-public}
 	% When $R$ is a constant, Problem \ref{eq:dist-devi-SC} can be solved in polynomial time. 
 	% \end{theorem}
 
 \begin{theorem}
 	The social-cost-minimizing public signaling scheme for \scg{}s, under optimistic equilibrium selection,  can be computed in $\poly(N, 2^{R(R-1)})$ time. 
 \end{theorem}

 \begin{proof} \,  The proof is divided into three major steps. 
 	
 	\subsubsection*{ Step 1: reducing optimal signaling to the best posterior problem.} 
 	
 	The first step transforms the optimal public signaling problem into its dual problem, which turns out to be easier to work with. We begin with a few convenient notations. For any congestion function $\C = \{ c_r \}_{r \in [R]} $, we use $NE^*(\C)$ to denote the social-cost-minimizing Nash equilibrium under congestion function $\C$ and $SC^*(\C)$ is the corresponding social cost of $NE^*(\C)$. 
 	% A public signaling scheme is said to be \emph{optimistically} optimal if it minimizes the sender's expected cost assuming the agents play the \emph{social-cost-minimizing} Nash equilibrium under any public signal.   
 	For any posterior   $\bv{p} \in \Delta_{\Theta} $, let $\C(\bv{p}) = \{ \sum_{\theta} p_{\theta} c_r^{\theta}(i) \}_{r \in [R]}$ denoted the expected congestion functions. As discuss in the preliminary section,  a public signaling scheme  decomposes  the prior distribution $\mu$ into a distribution over posterior distributions. Therefore, the optimal public signaling problem in our setting can be formulated as the following   linear program with infinitely many non-negative variables $\{x_{\bv{p}}>0 \}_{\bv{p} \in \Delta_{\Theta}}$: 
 	\begin{equation}\label{lp:public-primal}
 		\min [x_{\bv{p}} \times SC^*(\C(\bv{p})) ], \quad \text{s.t.} \quad   \sum_{\bv{p}} x_{\bv{p}} p_{\theta} = \mu_{\theta}, \, \,  \forall \theta \quad \text{ and } \quad \sum_{\bv{p}} x_{\bv{p}}   = 1  
 	\end{equation}
 	where $SC^*(\C(\bv{p}))$ is the social cost under any public signal with posterior $\bv{p} \in \Delta_{\Theta} $.   Through a duality argument, previous work by  \citet{ bhaskar2016hardness,mixture_selection} shows that solving LP \eqref{lp:public-primal} reduces to the following optimization problem for any weight parameter $\bv{w} \in \RR^{|\Theta|}$: 
 	\begin{equation}\label{eq:dist-devi-SC}
 		\min_{\bv{p} \in \Delta_{\Theta}} \bigg[    SC^*(\C(\bv{p})) + \bv{w} \cdot \bv{p}  \bigg] 
 	\end{equation} 
 	Note that Optimization Problem (OP) \eqref{eq:dist-devi-SC} tries to find the posterior $\bv{p} $  to minimize the social cost $ SC^*(\C(\bv{p}))$, with a slightly adjustment in the objective by a linear function of $\bv{p} $ with coefficient $\bv{w}$. Therefore, we conveniently refer it as the \emph{best weight-adjusted posterior} problem.   
 	
 	Problem \eqref{eq:dist-devi-SC} only provides a different (dual) perspective for   optimal public signaling, but certainly does not magically solve the problem. The key challenge for solving OP \eqref{eq:dist-devi-SC}  is  that $ SC^*(\C(\bv{p}))$ as the social cost at the optimistic equilibrium is not a convex function of $\bv{p} $. We provide an example to show this in Appendix \ref{append:social-cost-non-convex-example}. In fact, even computing the optimistic equilibrium  for any given posterior $\bv{p} $ is already quite non-trivial (see, e.g., \cite{ ieong2005fast} for a complex dynamic programming approach with $O(N^6R^5)$ running time), let alone optimizing the social cost at equilibrium by picking the best posterior $\bv{p} $. 
 	
 	\subsubsection*{ Step 2: a simple $O(R^N)$ time algorithm for solving  Problem \eqref{eq:dist-devi-SC}.} 
 	The main challenge of the proof is to solve the non-convex OP \eqref{eq:dist-devi-SC}. In this step, we describe a simple approach to solving OP \eqref{eq:dist-devi-SC}, which takes $O(R^N)$ time. We will then show how to accelerate the algorithm later. The key idea is to dividing  OP \eqref{eq:dist-devi-SC} into $O(R^N)$ many linear programs. The key observation here is that, though  OP \eqref{eq:dist-devi-SC} is non-convex, if we  constrain it to only the posterior distribution $\bv{p}$ that induce some action profile $\bv{a}$ as a pure NE, then we will obtain a linear program. Specifically, for any fixed action profile $\bv{a}$, the following linear program computes the best weight-adjust posterior, among all posteriors under which $\bv{a}$ is a Nash equilibrium.
 	\begin{lp}\label{lp:public-naive-lp}
 		\mini{ \sum_{r \in R} \sum_{\theta \in \Theta} p_{\theta} c_r^{\theta} (n_r) + \sum_{\theta} w_{\theta} p_{\theta} }
 		\st 
 		\qcon{\sum_{\theta} p_{\theta} c_{a_i}^{\theta} (n_{a_i}) \leq \sum_{\theta} p_{\theta} c_{r'}^{\theta} (n_{r'}+1)}{i \in [N], r' \in A_i}
 		\con{\sum_{\theta} p_{\theta} =1} 
 		\qcon{p_{\theta} \geq 0}{\theta \in \Theta}
 	\end{lp}
 	where configuration $\bv{n} = \bv{n}(\bv{a})$ is fixed due to the fixed action profile $\bv{a}$.  The first constraint guarantees that any agent $i$'s action $a_i$ is indeed a best response. This constraint also highlights the difficulty in handling asymmetric agent action sets $A_i$, leading to different incentive constraints for different agents. 
 	
 	Let $\bv{p}^* $  be the optimal solution to  LP \eqref{lp:public-naive-lp}.  Notably, LP \eqref{lp:public-naive-lp} only guarantees that the given $\bv{a}$ is a pure Nash equilibrium, but not necessarily the $NE^*(\C(\bv{p}^*))$, i.e., the optimistic social-cost-minimizing NE under posterior distribution $\bv{p}^*$. Nevertheless, the following lemma shows that we can obtain an optimal solution to OP \eqref{eq:dist-devi-SC}  by solving  LP \eqref{lp:public-naive-lp} separably for each action profile $\bv{a}$.  
 	
 	\begin{lemma}
 		An optimal solution to  OP \eqref{eq:dist-devi-SC}  can be obtained by solving  LP \eqref{lp:public-naive-lp} separably for each action profile $\bv{a} \in A$ and then picking the LP with the minimum objective. 
 	\end{lemma}
 	
 	\begin{proof} \, 
 		Let $LP(\bv{a})$ denote objective of LP \eqref{lp:public-naive-lp} w.r.t. a given $\bv{a}$. Let $\bv{a}^* = \min_{\bv{a} \in A} LP(\bv{a})$  be action profile with minimum objective among all these LPs (set any infeasible LP's objective  to be $\infty$), and $\bv{p}^*$ is the optimal solution to $LP(\bv{a}^*)$.   We argue that  $\bv{p}^*$ is an optimal solution to Problem \eqref{eq:dist-devi-SC} and moreover, $\bv{a}^*$ is the corresponding social-cost-minimizing Nash equilibrium under $\bv{p}^*$.
 		
 		The argument follows two observations.  Let $OPT$ denote the optimal objective of Problem \eqref{eq:dist-devi-SC}.    First, the social cost of equilibrium  $\bv{a}^*$  under posterior $\bv{p}^*$ is at most $OPT$. Specifically, let $\tilde{\bv{p}}$ be the optimal solution to Problem \eqref{eq:dist-devi-SC} and $\tilde{\bv{a}}$ be the optimistic NE under $\tilde{\bv{p}}$. By definition, $OPT$ is the social cost of equilibrium $\tilde{\bv{a}}$ under $\tilde{\bv{p}}$, plus $\bv{w} \cdot \tilde{\bv{p}}$. Instantiating LP  \eqref{lp:public-naive-lp} for $\bv{a} = \tilde{\bv{a}}$, we know $\tilde{\bv{p}}$ must be a feasible solution to LP \eqref{lp:public-naive-lp} since under $\tilde{\bv{p}}$, $\tilde{\bv{a}}$ is indeed a NE by definition and thus $\tilde{\bv{p}}$ satisfies the first constraint. Therefore, the objective of  equilibrium $\bv{a}^*$ under posterior $\bv{p}^*$ must be at most $OPT$.
 		
 		Second, the social cost of equilibrium $\bv{a}^*$  under posterior $\bv{p}^*$ is at least $OPT$. This is because $\bv{p}^*$ is a feasible solution to Problem \eqref{eq:dist-devi-SC} under which $ SC^*(\C(\bv{p})) + \bv{w} \cdot \bv{p}$ is at most the social cost of $\bv{a}^*$ plus $\bv{w} \cdot \bv{p}^*$. To conclude, they must be equal and thus imply  the lemma claims. 
 	\end{proof}

 	\subsubsection*{Step 3: accelerated $\poly(N, 2^{R(R-1)})$-time algorithm via equilibrium categorizations.}  
 	The crux of the proof is the third step which accelerates the simple algorithm in Step 2. The key idea underlying the algorithm in Step 2 is to ``divide'' the feasible region of OP \eqref{eq:dist-devi-SC},  i.e., $\Delta_\Theta$,   into a collection of many smaller regions; each region corresponds to an action profile $\bv{a}$, which is guaranteed to be an equilibrium for all $\bv{p}$ within its region. The optimization problem restricted to that region is precisely the LP \eqref{lp:public-naive-lp}. Unfortunately, there are too many action  profiles, which lead to too many LPs to solve and thus the inefficiency of the algorithm in Step 2.

 	Our key idea to overcome the above inefficiency is to divide the feasible region of OP \eqref{eq:dist-devi-SC} according to some different criteria, which: (1) can still lead to a tractable optimization program for each region (hopefully, an LP as well); (2) has much less number of regions and thus less optimization programs to solve.  % classify the action profiles into different \emph{categories} so that we will only need one linear program for \emph{each category}, as opposed to previous straightforward approach of using one linear program for each action profile $\bv{a}$. 
 	%How to properly categorize the action profiles is certainly the key challenge here and requires us to identify some carefully defined \emph{characteristic properties} of the equilibrium, which can be used for  categorizing  action profiles (i.e.,  all possible pure NEs).  
 	How to come up with the proper characteristics to  divide the feasible region of OP \eqref{eq:dist-devi-SC} is   the major challenge here.

 	The first thought one might have is to divide the feasible region  $\Delta_\Theta$ of posteriors based on the configuration $\bv{n}$ which summarizes the number of agents at each resource. Due to asymmetric agent action sets, it turns out that $\bv{n}$ does not contain sufficient information to describe the incentive constraints at equilibrium, like the first constraint of LP \eqref{lp:public-naive-lp}.  Our key idea is to divide the set of all posteriors into around $O(N^R 2^{R(R-1)})$ regions; each region  is uniquely determined by a configuration   $\bv{n}$ and, additionally,  $R(R-1)$ sign labels from $\{ \leq, > \} $ for any ordered pair of resource $(r,r')$. More concretely, any equilibrium leads to a configuration $\bv{n}$, which is a partition of the number $N$ into $R$ non-negative integers. Moreover,  another useful characteristics of any equilibrium is the ``deviation tendency'' from any resource $r$ to $r'$ --- i.e., whether an agent at resource $r$  has   incentives to deviate to any other $r'\not = r$ (regardless whether $r'$ is a feasible action or not). This can be checked by examining whether or not  \begin{equation}\label{eq:label-sign}
 		\sum_{\theta} p_{\theta} c_r^{\theta} (n_r) \leq \sum_{\theta} p_{\theta} c_{r'}^{\theta} (n_{r'}+1).
 	\end{equation}  
 	Therefore, we can associate each ordered pair $(r,r')$ with either a label ``$\leq$'' or ``$>$'' depends on the above inequality holds or not.   The characteristic properties we use to classify action profiles is precisely $(\bv{n}, \Lambda)$, in which $\Lambda \in \{ \geq , < \}^{R \times (R-1)}$ contains   the sign labels for all ordered pairs. We also call   $(\bv{n}, \Lambda)$ a \emph{signature} of  any equilibrium. % An equilibrium may   contain other information as well, e.g., which agent picks which resource. It just   turns out that this   choice of $(\bv{n}, \Lambda)$ suffices for our algorithm and has a relatively controllable set with 
 	Note that there are at most  $O(N^R 2^{R(R-1)})$ possible signature values.

 	% significantly reduce the total number of action profiles that one has to examine to be polynomial in $N$ by carefully classifying the $O(R^N)$ action profiles into a much smaller number of categories, specifically, $O(N^R 2^{R(R-1)})$ many categories. We will thus only need to solve  $O(N^R 2^{R(R-1)})$  many LPs, which can be down in polynomial time when $R$ is a constant.    
 	
 	There are several reasons that the signature $(\bv{n}, \Lambda)$ turns out to be a proper characteristics for categorizing the action profiles.  First, given any $\bv{n}$ and posterior $\bv{p}$, in must induce some characteristics $(\bv{n}, \Lambda)$ since the label of any $(r,r')$ resource pair can be directly checked by Equation \eqref{eq:label-sign}.   Therefore, $(\bv{n}, \Lambda)$ can be used as categorizing all posterior $\bv{p} \in \Delta_{\Theta}$ into different categories, without missing any of them. For convenience, we shall say any  $\bv{p}$ is categorized into some $(\bv{n}, \Lambda)$.

 	Second, for any signature $(\bv{n}, \Lambda)$,  we  can directly determine whether there exists a pure NE $\bv{a}$  that is ``consistent'' with $(\bv{n}, \Lambda)$. Moreover, this $\bv{a}$ will be a pure NE for \emph{any} posterior distribution categorized into $(\bv{n}, \Lambda)$ (since   $\Lambda$ has already contained all the incentive restrictions for agents). Formally, we  introduce  a useful notion of whether an equilibrium action profile $\bv{a}$ \emph{obeys} any given $(\bv{n}, \Lambda)$ or not. Suppose  $\bv{a}$   assigns agent $i$ to resource $r$, then by the definition of equilibrium, we know that agent $i$ does not have any incentive to deviate to any $r' \in A_i$. This means the $\Lambda$ induced by $\bv{a}$ must satisfy that $(r,r')$ has label $\leq$ for all $r' \in A_i$. If the $\Lambda$ satisfies the above requirements for all the assignment in the equilibrium profile $\bv{a}$ and moreover $\bv{n}(\bv{a}) = \bv{n}$,  we say equilibrium profile $\bv{a}$  obeys equilibrium signature $(\bv{n}, \Lambda)$. Note that the above notion only applies to equilibrium action profile $\bv{a}$ and has no meaning for non-equilibrium profile where deviation incentives are not present. % Second,   each equilibrium $\bv{a}$ will obey at least one, but maybe multiple,  $(\bv{n}, \Lambda)$. The second point is crucial since it implies if we categorize equilibrium based on $(\bv{n}, \Lambda)$, we will not miss any equilibrium action profile. 

 	While any equilibrium $\bv{a}$ obeys at least one signature $(\bv{n}, \Lambda)$, the reverse is \emph{not} true. That is, there exists  $(\bv{n}, \Lambda) \in P(A) \times \{ \geq , < \}^{R \times (R-1)}$ that does not correspond to the signature of any equilibrium $\bv{a}$. For instance, in a game with $R =2$ resources, if  both resource pair $(1,2)$ and $(2,1)$  has label ``$>$'', then this cannot be the signature of any equilibrium as agents at resource $1$ and $2$ both want to deviate.     
 	
 	\begin{figure}{}
 		\begin{center}
 			\begin{subfigure}[t]{0.45\textwidth}
 				\centering
 				\includegraphics[width=\linewidth]{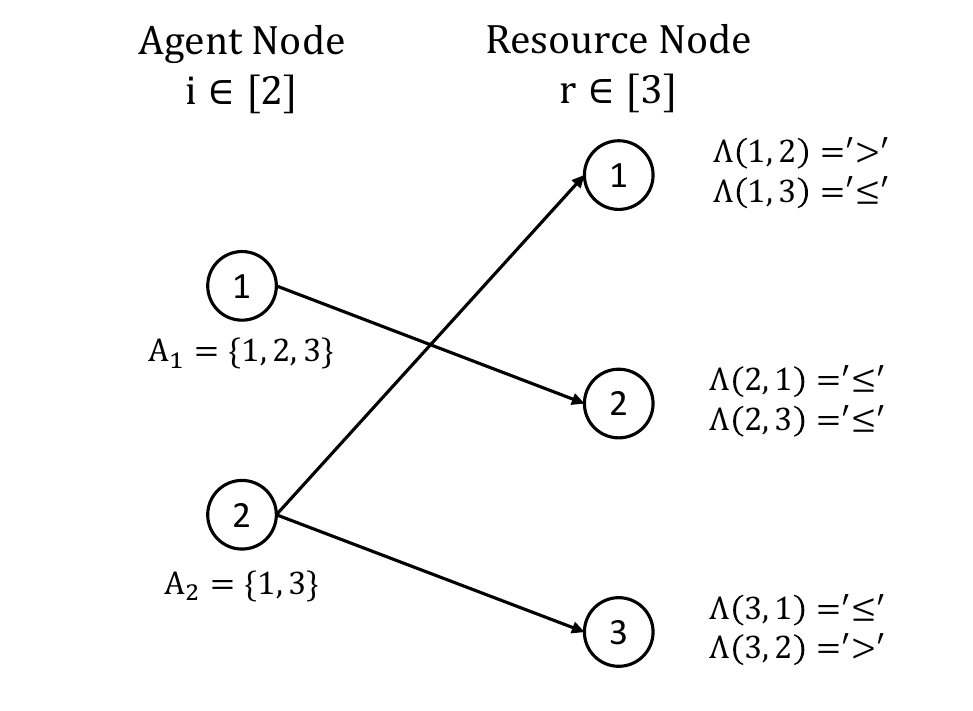}
 				\caption{An example with sign label $\Lambda$ depicted for each resource node,  action set for each agent, and corresponding equilibrium-allowable edges. For instance,  $1 \rightarrow 1$ is \emph{not} equilibrium-allowable because $\Lambda(1, 2)='>'$; $1 \rightarrow 2$ is equilibrium allowable because $\Lambda(2, 1)='\leq'$ and $\Lambda(2, 3)='\leq'$. %; $2 \rightarrow 2$ is not equilibrium-allowable since $2 \notin A_2$.
 				}
 				\label{fig:equilibrium-allowable-edge}
 			\end{subfigure}%
 			~~~~~~
 			\begin{subfigure}[t]{.45\textwidth}
 				\centering
 				\includegraphics[width=\linewidth]{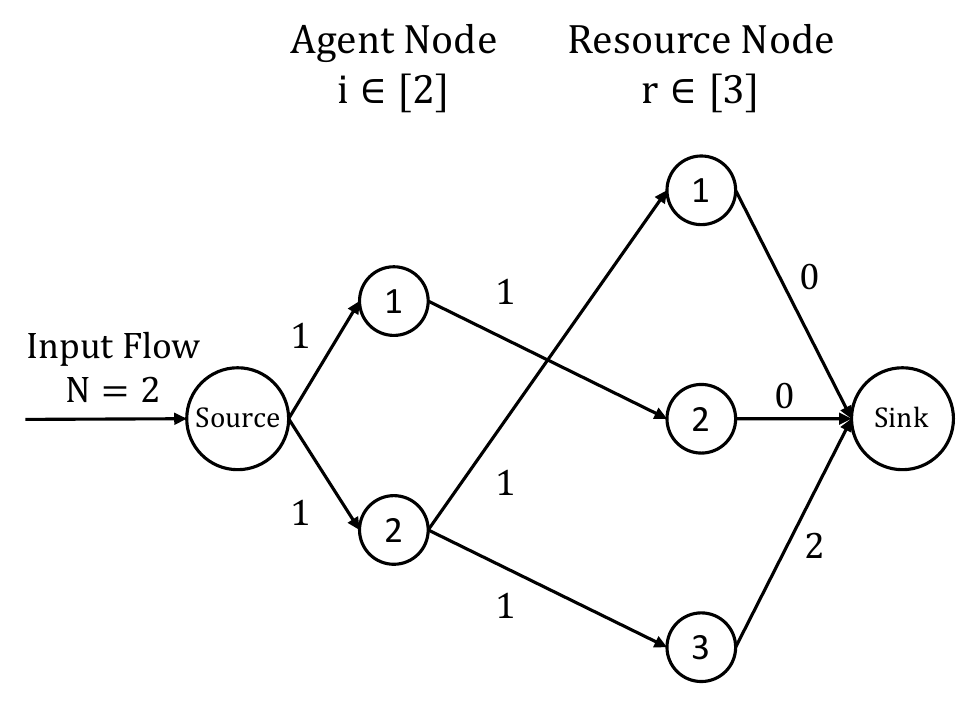}
 				\caption{The conversion of the equilibrium-allowable grapn in Figure \ref{fig:equilibrium-allowable-edge} into a flow problem with capacity depicted around each edge ($\bv{n}$ here is chosen as $(0,0,2)$). %graph construction on bipartite graph $G$. Input flow equals $N$; $capacity = 1$ for all edges from source node to agent nodes and all equilibrium allowable edges; $capacity = n_r$ for all edges from resource node $r$ to sink node.
 				}
 				\label{fig:equilibrium-allowable-max-flow}
 			\end{subfigure}
 			\caption{An example graph with  only  equilibrium-allowable edges (Left) and the conversion of the graph into a flow problem for identifying an equilibrium (Right). }
 			\label{fig:equilibrium-allowable-illust}
 		\end{center}
 	\end{figure}

 	We now provide an  efficient algorithm to determine, for any given signature $(\bv{n}, \Lambda)$,  whether it is possible to have an equilibrium action profile $\bv{a}$ that \emph{obeys} the $(\bv{n}, \Lambda)$. Given any $(\bv{n}, \Lambda)$, an equilibrium action profile $\bv{a}$ can possibly assign agent $i$ to resource $r$ only if $(r,r')$ has label $\leq$ in $\Lambda$ for all $r' \in A_i$, which means $i$ has no incentive to deviate to any other resource $r'$ in its feasible action set $A_i$.  In this situation, we say assignment $i \to r$ is \emph{equilibrium-allowable} in signature $(\bv{n}, \Lambda)$. For any $(\bv{n}, \Lambda)$, we can identify all the equilibrium-allowable assignments as a bipartite graph $G = ([N],[R], E)$ where $e = (i \to r)$ is an edge if and only if assignment $i \to r$ is equilibrium-allowable. Figure \ref{fig:equilibrium-allowable-edge} illustrates this construction with a concrete example.  It is easy to verify that in  Figure \ref{fig:equilibrium-allowable-edge}      resource $1$ is equilibrium allowable for agent $1$ while   resource $1$ and $3$ are equilibrium allowable for agent $2$.   To find an equilibrium that obeys the given $(\bv{n}, \Lambda)$, we only need to match all agents in $[N]$ to resource set $[R]$ using only equilibrium-allowable nodes, with an additional constraint that resource $r$ is mapped to exactly $n_r$ agents. This is a   one-to-many bipartite matching problem with fixed demand on the right hand side. Whether such a matching exists or not can be solved by a standard max-flow formulation (for example, there is no matching for $(\bv{n}, \Lambda)$ in Figure \ref{fig:equilibrium-allowable-illust} since the max flow in Figure \ref{fig:equilibrium-allowable-max-flow} equals 1, which is less than $N=2$). Notably, any feasible matching will be a pure NE that obeys the given $(\bv{n}, \Lambda)$ since the equilibrium constraint is directly imposed by  $\Lambda$.  
 	
 	Finally, we observe that the social cost of any NE that obeys any given $(\bv{n}, \Lambda)$ is $\sum_{r \in R} \sum_{\theta \in \Theta} p_{\theta} c_r^{\theta} (n_r) $. Now consider any $(\bv{n}, \Lambda)$ such that there exists a pure NE obeying it (efficiently decidable via the aforementioned matching algorithm). We claim that the following linear program computes the best weight-adjusted posterior, among \emph{all posteriors that induce} $(\bv{n}, \Lambda)$.  
 	
 	\begin{lp}\label{lp:public-efficient-lp}
 		\mini{ \sum_{r \in R} \sum_{\theta \in \Theta} p_{\theta} c_r^{\theta} (n_r) + \sum_{\theta} w_{\theta} p_{\theta} }
 		\st 
 		\qcon{\sum_{\theta} p_{\theta} c_{r}^{\theta} (n_{r}) \leq \sum_{\theta} p_{\theta} c_{r'}^{\theta} (n_{r'}+1)}{  (r,r') \text{ with } \Lambda(r,r') = ``\leq'' }
 		\qcon{\sum_{\theta} p_{\theta} c_{r}^{\theta} (n_{r}) \geq \sum_{\theta} p_{\theta} c_{r'}^{\theta} (n_{r'}+1)}{  (r,r') \text{ with } \Lambda(r,r') = ``>'' }
 		\con{\sum_{\theta} p_{\theta} =1} 
 		\qcon{p_{\theta} \geq 0}{\theta \in \Theta}
 	\end{lp}
 	Note that the first two constraints guarantees that $\Lambda$ is satisfied for any feasible $\bv{p}$. The only caveat here is that when $\Lambda(r,r') = ``>''$, we used the ``$\geq$'' nevertheless. This is fine since even if the ``='' holds, an agent at $r$ still will not deviate to $r'$ (just as if ``$>$'' holds) since it is a tie.   
 	
 	Consequently, OP \eqref{eq:dist-devi-SC} can be solved by solving LP \eqref{lp:public-efficient-lp} for every $(\bv{n}, \Lambda)$ such that there exists a pure NE obeying it, and then picks the one with the smallest objective value. One small caveat here is that for any optimal solution $\bv{p}^*$ to LP \eqref{lp:public-efficient-lp}, the equilibrium with $\bv{n}$ as configuration may not be the social-cost minimizing equilibrium. Similar to our argument at the end of Step 2, this will not be an issue for the special LP with the minimum objective. This concludes our proof of the theorem.

 \end{proof}

 	\subsection{Optimal private signaling.}
 	
 	We now consider   optimal private signaling. Our starting point is a celebrated characterization by Bergemann and Morris \citet{ Bergemann16Bayes} that all the Bayes Nash equilibria that can possibly arise at any private signaling scheme forms the set of \emph{Bayes correlated equilibrium} (BCEs). Based on the characterization, we first formulate a  linear program that computes the social-cost-minimizing BCEs. Unfortunately, this linear program has $\Omega(R^N)$ many variables since it  has to  enumerate all possible action profiles $\bv{a}$.
 	% , a  straightforward way to compute a social-cost-minimizing private signaling scheme is to formulate the problem as a linear program with variables $\{\pi(\mathbf{a}|\theta)\}$. However, this linear program cannot be solved efficiently due to an exponential number of possible action profiles $\mathbf{a}$ (aka., signals).  
 	In order to design an efficient algorithm, our idea is to identify compact yet still sufficiently expressive   marginal probabilities   to capture the necessary information needed for each agent's inference under a signaling scheme. This ``interim'' description of the signaling scheme is   often called the ``reduced form'' in auction design \cite{ cai2012algorithmic,alaei2012bayesian,border1991implementation} and recently in information design \cite{ candogan2020reduced}. Specifically,   we  introduce the marginal variables $x_{\theta \bv{n} i r}$ to denote the probability that  agent $i$ is recommended to resource $r$ and the resulting configuration is $\textbf{n}$, conditioned on the state of nature $\theta$. We show that  $\{ x_{\theta \bv{n} i r} \}$ suffices to characterize agent $i$'s inferences on the uncertainty in the game, including the state of nature, other agents' behavior and her own utility. In addition, we develop a novel technique to sample a private signaling scheme from our constructed ``interim'' marginal probabilities. This sampling technique strictly generalizes a classic result in statistics by \citet{ tille1996elimination} and may be of independent interest.   
 	%,  must capture each agent's expectation of their own action, as well as the outcome of other agents' choices in a compact manner. 
 	Our main theorem is stated as follows.
 	
 	\begin{theorem}\label{thm:opt-private}
 		The social-cost-minimizing private signaling scheme can be computed in $\poly(N^R)$ time. 
 	\end{theorem}

 	% We therefore introduce the marginal variables $x_{\theta \bv{n} i r}$ to denote the probability that conditioned on the state of nature $\theta$, agent $i$ is recommended to resource $r$ and the resulting configuration is $\textbf{n}$. We'll show that  $x_{\theta \bv{n} i r}$ is necessary and sufficient to characterize agent $i$'s belief about every aspect of uncertainty in the game, including the state of nature, her own action, and other agents' behavior. In addition, we develop a sampling technique to obtain a private signaling scheme from our constructed ``interim'' belief, which strictly generalizes a classic result in statistics \citet{ tille1996elimination} and may be of independent interest.

 	\begin{proof}  [Proof Sketch] \,  
 		We provide a proof sketch here and defer the detailed proof to Appendix \ref{append:const-private-proof}.  \citet{ bergemann2016bayes} show that we only need to optimize the social cost over the set of \emph{Bayes Correlated Equilibrium} (BCEs).  Like the standard correlated equilibria, the signals of a private signaling scheme in a BCE can be interpreted as \textit{obedient action recommendations}. Thus, the signal set $\Sigma_i$ for agent $i$ can W.L.O.G. be $A_i$, and $\Sigma = A$. An action recommendation $a_i$ to player $i$ is \textit{obedient} if following this recommended action is indeed a best response for $i$. Formally, for any $a_i$, we have: 
 		\begin{equation} \label{model:eq:obediant_constraint}
 			\sum_{\theta \in \Theta} \mu_{\theta} \sum_{\bv{a}_{-i}} \pi(a_i, \bv{a}_{-i}|\theta) c^{\theta}_{a_i}(n_{a_i}) \leq \sum_{\theta \in \Theta} \mu_{\theta} \sum_{\bv{a}_{-i}} \pi(a_i, \bv{a}_{-i}|\theta) c^{\theta}_{a'_i}(n_{a'_i}+1), \quad \forall a'_i \in A_i
 		\end{equation}
 		Consequently, the optimal correlated equilibrium can be computed by the exponentially large linear program:  
 		\begin{lp}\label{lp:private-naive}
 			\maxi{ \sum_{\theta \in \Theta} \sum_{\bv{a} \in A} \sum_{r \in R} \pi(\bv{a}|\theta) \mu(\theta) n_r c^\theta_r(n_r) }
 			\st 
 			\qcon{   \sum_{\theta \in \Theta} \sum_{\bv{a}_{-i}} \pi(a_i, \bv{a}_{-i}|\theta) \mu_{\theta}  \bigg[ c^{\theta}_{a_i}(n_{a_i}) -  c^{\theta}_{a'_i}(n_{a'_i}+1) \bigg] \leq 0}{ i \in [N],   a'_i \in  A_i}
 			\qcon{   \sum_{\bv{a} \in A}  \pi(  \bv{a} |\theta) =1}{   \theta \in \Theta}
 			\qcon{ \pi(  \bv{a} |\theta) \geq 0 }{\theta \in \Theta, \bv{a} \in A}
 		\end{lp}

 		To efficiently solve LP \eqref{lp:private-naive},  we define the variable of marginal probability   $x_{\theta \bv{n} i r} $  to denote the probability that conditioned on the state of nature is $\theta$, agent $i$ is assigned to resource $r$ and the configuration of all resources is $\bv{n}$. That is,
 		\begin{equation*}
 			x_{\theta \textbf{n} i r} = 
 			\sum_{\textbf{a} \in A} \pi(\textbf{a}|\theta) \cdot \mathbb{I}(a_i = r, \bv{n}(\bv{a}) = \bv{n} )\end{equation*} 
 		We conveniently refer to $\{ x_{\theta \bv{n} i r} \}$ as the \emph{reduced form} of the private signaling scheme $\{ \pi(\bv{a}|\theta) \}$.  Utilizing the above definition, the key of our proof is to argue that the linear program in Figure \ref{fig:lp-const-private} exactly computes the optimal private signaling scheme. 
 		\begin{figure}[ht]
 			\centering
 			\scalebox{1}{
 				\begin{tikzpicture}[
 					farnodes/.style={
 						node distance=1cm
 					},
 					model1/.style = {
 						rectangle, 
 						text width=14.5cm, 
 						minimum height=1cm,
 					}
 					]
 					\draw[fill=black, opacity=0.1] (16, -7.5) rectangle (0,0);
 					
 					\node[farnodes] (DLE) [model1] at (7.5, -0.3){\textbf{Variables:}};
 					
 					\node[farnodes](va) [model1] at (8.5, -1.2){
 						$x_{\theta \bv{n} i r}, \forall \theta \in \Theta, \bv{n} \in P(A), i \in [N], r \in [R]$ \\
 						/\  \texttt{marginal probability that configuration is $\textbf{n}$ and $i$ receives recommendation $r$, conditioned on state $\theta$.}
 					};
 					
 					\node[farnodes] (DLE) [model1] at (7.5, -2.3){\textbf{Minimizing:}};
 					
 					\node[farnodes](va) [model1] at (8.5, -3){
 						$\sum_{\theta} \mu_{\theta} \sum_{\bv{n} \in P(A)} \sum_{i \in [N]} \sum_{r \in [R]} x_{\theta \bv{n} i r} \cdot c^{\theta}_{r}(n_{r}) $ \\
 						%  \texttt{/\  Expected social cost when agents follows recommendation.}
 					};
 					
 					\node[farnodes] (DLE) [model1] at (7.5, -3.8){\textbf{Constraints:}};
 					
 					\node[farnodes](va) [model1] at (8.0, -5.5){
 						\begin{lp*} \vspace{1mm}
 							\qcon{\sum_{\theta} \mu_{\theta} \sum_{\bv{n} \in P(A)} x_{\theta \bv{n} i r} \left[ c^{\theta}_{r}(n_r) - c^{\theta}_{r'}(n_{r'} + 1)  \right] \leq 0}{ \forall i \in [N], r, r' \in A_i }\vspace{1mm}
 							\qcon{ x_{\theta \bv{n} i r} = 0 }{ \forall i \in [N], r \notin A_i, \bv{n} \in P(A), \theta \in \Theta } \vspace{1mm} 
 							\qcon{ \sum_{\textbf{n} \in P(A)} \sum_{r \in [R]} x_{\theta \textbf{n} i r} = 1  }{  \forall i \in [N], \theta \in \Theta } \vspace{1mm}
 							\qcon{ \sum_{j \in [N]}  x_{\theta \textbf{n} j r} = n_r\sum_{r' \in [R]} x_{\theta \textbf{n} i r'}  }{\forall i \in [N],  r \in [R], \theta \in \Theta, \textbf{n} \in P(A)}  
 							\qcon{ x_{\theta \bv{n} i r} \geq 0 }{ \forall i \in [N], r \in [R], \textbf{n} \in P(A), \theta \in \Theta  }
 						\end{lp*}
 						%     $\sum_{\theta} \mu_{\theta} \sum_{\bv{n} \in P(A)} x_{\theta \bv{n} i r} \left[ c^{\theta}_{r}(n_r) - c^{\theta}_{r'}(n_{r'} + 1)  \right] \leq 0, \, \qquad  \forall i \in [N], r, r' \in A_i$, \vspace{3mm}\\
 						%   % /\ \texttt{Obedience constraint} \\ \vspace{1mm} 
 						
 						%     $x_{\theta \bv{n} i r} = 0, \qquad \qquad \qquad  \forall i \in [N], r \notin A_i, \bv{n} \in P(A), \theta \in \Theta$ \\ \vspace{1mm} 
 						
 						%     $\sum_{\bv{n} \in P(A)} \frac{1}{ n_r }  \sum_{j \in [N]} x_{\theta \bv{n} j r} = 1, \forall r \in [R], \theta \in \Theta$ \\ \vspace{1mm} 
 						
 						%     $\frac{1}{ n_r }  \sum_{j \in [N]}  x_{\theta \textbf{n} j r} =\sum_{r' \in [R]}x_{\theta \textbf{n} i r'} , \forall  i \in [N],  r \in [R], \bv{n} \in P(A), \theta \in \Theta$ \\ \vspace{1mm} 
 						
 						%     \item $x_{\theta \bv{n} i r} \geq 0, \forall i \in [N], r \in [R], \textbf{n} \in P(A), \theta \in \Theta$
 					};
 				\end{tikzpicture}
 			}
 			\caption{Linear Programming Formulation for Social Cost Minimizing Private Signaling Scheme}
 			\label{fig:lp-const-private}
 		\end{figure}
 		
 		% \begin{align*}
 			%     \sum_{\textbf{n} \in P(A)} \sum_{r \in [R]} x_{\theta \textbf{n} i r} = 1, \quad  \forall i \in [N], \theta \in \Theta  \label{const:lp-feasibility-2} \\
 			%             \sum_{j \in [N]}  x_{\theta \textbf{n} j r} = n_r\sum_{r' \in [R]} x_{\theta \textbf{n} i r'} , \quad & \forall i \in [N],  r \in [R], \theta \in \Theta, \textbf{n} \in P(A)  \label{const:lp-feasibility-3}
 			% \end{align*}

 		It can be verified that the objective and the first constraint of the LP in Figure \ref{fig:lp-const-private}  are equivalent to that in LP \eqref{lp:private-naive}. Key to our proof is to argue that the remaining constraints exactly characterize the set of all  reduced forms that can be induced by some private signaling scheme. 
 		
 		We start by illustrating why these constraints are necessary. First,  for any $i$,  we have  $x_{\theta \bv{n} i r} = 0$ for any  $  r \notin A_i, \bv{n} \in P(A), \theta \in \Theta$  since agent $i$ cannot be allocated to any $r\notin A_i$. Second, $\sum_{\textbf{n} \in P(A)} \sum_{r \in [R]} x_{\theta \textbf{n} i r} = 1, \forall i \in [N], \theta \in \Theta$. This is because, $\sum_{r \in [R]} x_{\theta \textbf{n} i r}$ is essentially the probability that the configuration is $\mathbf{n}$ and agent $i$ is sent to a resource given $\theta$. This probability is thus equal to the probability that the configuration is $\mathbf{n}$ given $\theta$. This is because for any configuration, every agent $i$ is always recommended to one of the resources (i.e., $\sum_n n_r = N$). Summing over this probability over all $\bv{n} \in P(A)$ should equal to $1$.
 		
 		Finally, we have constraint $\sum_{j \in [N]}  x_{\theta \textbf{n} j r} = n_r\sum_{r' \in [R]} x_{\theta \textbf{n} i r'} ,  \forall i \in [N],  r \in [R], \theta \in \Theta, \textbf{n} \in P(A)$. In this equation, the LHS, $\sum_{j \in [N]}  x_{\theta \textbf{n} j r}$, is essentially $n_r$ times the probability that the configuration is $\mathbf{n}$ and an agent is assigned to resource $r$ given $\theta$. Intuitively, this is because there are $n_r$ agents assigned to resource $r$ in the configuration $\mathbf{n}$ and therefore, the sum of probabilities over agents will be $n_r$ times the probability an agent is assigned to resource $r$.

 		The crux of our proof is to show that the aforementioned sets of constraints on $\{ x_{\theta \bv{n} i r} \}$   exactly suffice to characterize  all feasible reduced forms. This is argued through a constructive proof. That is, given any $\{ x_{\theta \bv{n} i r} \}$ satisfying the constraints in LP of Figure \ref{fig:lp-const-private}, we design an efficient algorithm that samples a  private signaling scheme inducing  $\{ x_{\theta \bv{n} i r} \}$. Our flow-decomposition-based sampling technique also strictly generalizes a classic sampling procedure by \citet{ tille1996elimination}, which corresponds to the special case with $R = 2$.   
 		
 		Our sampling process has two steps: (1) sample a configuration $\bv{n}$ with probability $\sum_{r \in [R]} x_{\theta \textbf{n} i r}$; (2) sample $\bv{a} \in \{ \bv{a}: \bv{n}(\bv{a}) = \bv{n} \}$.  Step (2) is more involved since we have to efficiently sample from an action profile space with size exponential in $N$.  We highlight the key ideas next.  After a configuration $\bv{n}$ is sampled, we can compute $\PP(i \to r| \bv{n}, \theta) = x_{\theta \bv{n} i r} / \sum_{r' \in [R]} x_{\theta \bv{n} i r'}$ as the marginal probability that agent $i$ is assigned to resource $r$ conditioned on  state $\theta$ and configuration $\bv{n}$. The LP constraints imply  $\sum_{i \in [N]} \PP(i \to r| \bv{n}, \theta) = n_r$. Therefore, we can interpret $ \{  \PP(i \to r| \bv{n}, \theta) \} $ as a \emph{fractional flow} on a bipartite graph with left-side nodes as agents $[N]$ and right-side nodes as resources $[R]$.  The flow amount from agent $i$ to resource $r \in A_i$ is $ f(i \to r) = \PP(i \to r| \bv{n}, \theta)$  (the flow from $i$ to $r \not \in A_i$ is $0$); the total supply of flow going outside from agent $i$ is $1$ and the total demand of flow entering resource $r$ is $n_r$. Notably, since both the supply and demand are integers,   any feasible integer flow corresponds precisely to a deterministic action profile $\bv{a}$ with $\bv{n}(\bv{a}) = \bv{n}$. Thus, by decomposing the fractional flow $ \{  \PP(i \to r| \bv{n}, \theta) \} $ into a distribution over feasible integer flow (e.g., using Ford-Fulkerson), we efficiently generate a private signaling scheme $\pi(\bv{a}|\theta)$ that induces any feasible $\{ x_{\theta \bv{n} i r} \}$. 
 		
 		% This decomposition procedure can be done by running Ford-Fulkerson on a fractional flow graph. In each iteration, we find an integer max flow $f_1$ and re-set the flows of all the edges $e$ in $f_1$ to be $f(e) \leftarrow (f(e) - p_1 f_1(e))$, where $p_1 = \min_{e} \frac{f(e)}{f_1(e)}$ is the fraction of minimum flow in $f_1$. Note that the new flow still satisfies flow conservation constraints. In particular, any edge $(s,i)$ from source $s$ to agent node $i$ must have flow amount $(1 - p_1)$ since the max integer flow $f_1$ must flow $1$ unit flow through all these edges. Similarly, any edge $(r,t)$ from resource node $r$ to sink $t$ must have flow amount $(1- p_1) n_r$ since the max integer flow $f_1$ must flow $n_r$ unit flow through all these edges. We claim that the max-flow in the subgraphs with the new flow amount is still $N$ since dividing the current flow by $(1-p_1)$ gives rise to another fractional $N$-unit flow. We thus can continue this decomposition. This procedure will end in at most $N \times R + N + R$ rounds since in each iteration, at least one edge $e^* = \arg \min_{e} \frac{f(e)}{f_1(e)}$ is updated with flow $0$ and will be deleted from the graph. When it terminates, we obtain a distribution over integer flows, where each integer flow corresponds to precisely an action profile determined by the selected edges $(i, r)$. This gives rise to a feasible $\pi$ that induces $X$.
 		
 		Therefore,  the optimal private signaling scheme can be computed efficiently by solving the LP of Figure \ref{fig:lp-const-private} and then sample the optimal private signaling scheme. The total running time is polynomial in $N^R$, i.e., upper bound number of configurations.   
 		
 	\end{proof} 
 	
 	\begin{remark}
 		Familiar audience may notice that the above flow technique bear some conceptual similarity to the flow characterization of reduced form for auction design \cite{che2013generalized}. This is indeed true because both approaches try to capture the relation between marginal probabilities and the underlying full action or type profiles. However, this conceptual connection does not easily imply that the reduced form characterization for auction design can be directly applied to signaling. This is because the reduced form for signaling differs from the reduced form for auction design due to the different sets of marginal probabilities that each problem has to keep track of. Specifically, signaling schemes need to keep track of not only the marginal probabilities of each receiver's belief about each state type but also her belief about other receivers' actions since a receiver's incentives of action deviation are affected by both (this is why the reduced form  $x_{\theta \bv{n} i r}$ in the  proof of Theorem 2  includes receiver $i$'s belief about the profile of all other receivers' actions, summarized into  $\bv{n}$). However, in auction design, the auctioneer only needs to keep track of each bidder's belief about her own types since a bidder's misreport incentives are only affected by her own types' marginal allocation probabilities. This is a fundamental difference between the two problems, arising from their different game structures. Notably, this difference is significant. For example, an earlier work by \citet{Dughmi2016} shows that an efficient characterization of reduced form for polynomially many independent yet non-identical receivers is unlikely to exist unless the polynomial hierarchy collapses, whereas in the analogous auction setting with independent yet non-identical bidders, the celebrated work of \citet{border1991implementation} leads to an  efficient  characterization of the reduced form for auction design \cite{cai2012algorithmic,alaei2012bayesian}.  
 	\end{remark}
 	
 	\section{Hardness of symmetric \texttt{SCG}s with many resources   }
 	% In this section, we move to the setting with many locations and show that the problem become intractable generally.  

 	% The algorithm applies the ellipsoid algorithm to Problem (\ref{cce:dual-LP}). At each iteration, we require that the vector of dual variables $\alpha$ given to the separation oracle is player-symmetric, which can be obtained by applying the symmetrization technique introduced in the proof of Lemma \ref{cce:symmetrization}. The separation oracle needs to compute $\texttt{sep}(L', N')$ for all $L' \in [L]$  and $N' \in [N]$ through a dynamic programming approach. By comparing the optimal $\mathcal{X}(\theta)$ for every $\theta \in \Theta$, we can derive optimal dual in polynomial time. The cutting planes generated during the Ellipsoid method can be used to compute the optimal ex-ante signaling scheme with polynomial-sized support and the maximum social welfare under an incentive compatible ex-ante signaling scheme equals to the optimal dual value by LP duality.

 	In this section, we show that the restriction to a small number of resources in the previous section is   necessary for efficient algorithms.  Indeed, both public and private signaling exhibits intractability once we move to the general setup with many resources, even when agents have symmetric action spaces. % Note that,  symmetry indeed helps to simplify the action space. Under symmetry, any configuration $\bv{n}(\bv{a})$ contains (essentially) all information that $\bv{a}$ contains since it does not matter which agent picks a resource $r$ so long as there are $n_r$ agents. However, without symmetry, configuration $\bv{n}(\bv{a})$  contain significantly less information than the action profile $\bv{a}$ since  which agent picks which resource matters. Unfortunately, it turns out that the above simplification still does not suffice for us to obtain an efficient algorithm when there are many resources. 
 	
 	%tractability results from previous sections are essentially the best one can hope for. 
 	
 	\subsection{Equilibrium-oblivious intractability of public signaling.}\label{sec:public-many}
 	One challenge of proving hardness for optimal public signaling is  the possible existence of multiple equilibria. Therefore, the hardness under one equilibrium selection rule may not imply any clue about the hardness of another equilibrium choice. To address this issue, we introduce a stronger notion of hardness which captures intractability \emph{regardless of what equilibrium one chooses under any public signal}. 
 	
 	\begin{definition}[Equilibrium-Oblivious  Inapproximability]\label{def:oblivious-hard}
 		We say  it is   \emph{equilibrium-obliviously}  NP-hard to obtain an  $\alpha$-approximation for optimal public signaling   if it is NP-hard to compute a public signaling scheme $\pi$ such that its equilibrium social cost is at most $\alpha$ times of the equilibrium social cost of the optimal public signaling $\pi^*$ even when: 
 		\begin{itemize}
 			\item    the social cost of $\pi$ is   evaluated at the socially-best (i.e., cost-minimizing) Nash equilibrium;  whereas 
 			\item  the social cost of $\pi^*$ is evaluated at the socially-worst (i.e., cost-maximizing) Nash equilibrium.  
 		\end{itemize}
 		% In this case, we say optimal public signaling is equilibrium-obliviously  $\alpha$-inapproximable, or obliviously.
 		% $\alpha$-inapproximable for short.   
 	\end{definition}
 	When it is clear from the context, we simply say oblivious inapproximability or obliviously NP-hard. 
 	Oblivious  $\alpha$-inapproximability means it is intractable to obtain an $\alpha$-approximation  even when we favor the algorithm with the best equilibrium choice but defy the benchmark with the worst equilibrium choice.   This fully rules out any possibility of leveraging equilibrium selection to get a good approximation and thus is a firm hardness evidence irrespective of equilibrium selection.  Note that the  optimal public signaling here is the one that minimizes its social cost w.r.t. its  socially-worst  equilibrium choice.      
 	
 	% \thanh{We lack a transition to this section. Can we say anything about the NP-hardness of multiplicative constant approximate algorithm for the welfare maximizing optimal (ex-ante) private signaling?}
 	
 	Under Definition \ref{def:oblivious-hard}, the approximation ratio $\alpha$ can be smaller than $1$ due to the different equilibrium selection for the algorithm and the benchmark. Nevertheless, our following result shows that it is obliviously NP-hard to obtain a $(1+\frac{1}{5N})$-approximation algorithm and thus rules out  FPTAS for the social-cost-minimizing optimal public signaling in \scg{}s, irrespective of equilibrium selection rules.    
 	
 	\begin{theorem}\label{thm:hardness-public}
 		It is equilibrium-obliviously NP-hard to obtain a $(1+\frac{1}{5N})$-approximation algorithm for the social-cost-minimizing optimal public signaling in \scg{}s, even when agents have symmetric action sets.   
 	\end{theorem}
 	
 	% \subsubsection*{Proof of Theorem \ref{thm:hardness-public}. }
 	\begin{proof}   [Proof Sketch] \, 
 		One natural idea for proving equilibrium-oblivious hardness would be to construct games that always admit a unique Nash equilibrium, and thus we do not need to worry about the ``obliviousness'' part. Unfortunately, it turns out that for \scg{}s, it is extremely challenging (if not impossible) to construct games with a unique equilibrium under an arbitrary public signal. Our proof thus takes a different route --- we construct a class of games and then derive the upper or lower bounds for the social cost  of \emph{arbitrary} equilibrium by analyzing only the incentives at equilibrium.

 		Specifically, our reduction is from the following NP-hard problem.  \citet{ Khot2012} prove that for any positive integer $k$, any integer $q$ such that $q \geq 2^k+ 1$, and an arbitrarily small constant $\eps > 0$,  given an undirected graph $G$, it is NP-hard to distinguish between the following two cases:
 		\begin{itemize}
 			\item {\bf Case 1}: There is a $q$-colorable induced subgraph of $G$ containing a  $(1-\eps)$ fraction of all vertices, where each color class contains a $ \frac{1-\eps}{q}$ fraction of all vertices.
 			\item{\bf Case 2}: Every independent set in $G$ contains less than a $\frac{1}{q^{k+1}}$ fraction of all vertices.
 		\end{itemize} 
 		Without loss of generality, we assume that  no nodes in $G$ are adjacent to all other nodes, since the maximum independent set should never contain any such node. Otherwise, this is the only node that the independent set can contain. 
 		
 		Given a graph $G$ with vertices $V = \set{1,\ldots,R}$ and edges $E$, we will construct a public persuasion instance so that any desired algorithm for approximating the optimal sender utility can be used to distinguish these two cases.  There are $N =  \frac{1-\epsilon}{q}R$ agents and $R + 1$ resources which correspond to the $R$ nodes of the graph $G$, plus a ``backup'' resource $0$, i.e., $[R] = \{ 0 \} \cup V = \{ 0, 1, 2, \cdots, R \}$. The game is symmetric so all agents have the same action set $[R]$. The set of the states of nature $\Theta = V$ corresponds to vertices of the graph as well. The prior distribution is  uniform over states of nature --- i.e.,   $\theta \in V$ is realized with probability $1/R$. %We construct a ``backup'' resource $0$ with cost function $c_0^{\theta} (n) = 2N$ for any $\theta, n$ and choosing resource $0$. Any agent can choose resource $0$ and guarantee a cost of   $2N$. 
 		The congestion function of resource $r \in V$ at   state of nature $\theta$ is defined  as follows: 
 		\begin{itemize}
 			\item If $r = \theta$, let $c_{r}^{\theta}(n) = 1 - \frac{1}{n^2}, \forall n > 0$ and $c_{r}^{\theta}(0) = 0$. We call $r$ the \emph{good} resource. 
 			\item If $(r,\theta) \in E$ is an edge in $G$, let $c_{r}^{\theta}(n) = 3, \forall n \geq 0$. We call such an $r$ \emph{bad} resource. 
 			\item If $r( \not = \theta)$ is not adjacent to $\theta$ in graph $G$, let $c_{r}^{\theta}(n) = 1, \forall n \geq 0$. We call such an $r$ \emph{normal} resource. 
 			\item $c^{\theta}_0(n) = 1, \forall n \geq 0, \theta \in V$. Resource $0$ ensures that each agent  suffers cost at most $1$ and is   referred to as the \emph{backup} resource.
 		\end{itemize}
 		
 		The principal would like to minimize the total cost.  We show that a desired approximation to cost of   optimal public signaling scheme can help us to distinguish the two cases above. %, which thus implies the NP-hardness of approximating the social-cost-minimizing public signaling scheme.  
 		Specifically, the following Lemma~\ref{lemma.1} shows that in \textbf{Case 1}, optimal public signaling can obtain   social cost at most $N  - (1-\epsilon)$. On the other hand,   Lemma~\ref{lemma.2} shows that in \textbf{Case 2}, the expected social cost from any public signal is at least $ N  - \frac{3}{4} - \frac{1}{(1-\epsilon)q^k}    $. Since the hardness of the independent set instance holds for any parameter $\epsilon, q, k$, we can simply choose them to satisfy $\frac{1}{(1-\epsilon)q^k} \leq \frac{1}{40}$ and $\epsilon \leq \frac{1}{40}$, which will lead to $ N  - \frac{31}{40} $ as the lower bound of the social cost of the optimal public signaling scheme for \textbf{Case 2}.  This implies that any $(1+\frac{1}{5N})$-approximate algorithm to our constructed instance will be able to distinguish {\bf Case 1} and {\bf Case 2} --- specifically, it will output a solution with cost at most $\big[ N - (1-\epsilon)  \big] \times (1+\frac{1}{5N}) < N - \frac{4}{5} + \epsilon \leq  N - \frac{31}{40}$ for {\bf Case 1}  but output a solution with cost at least $N - \frac{31}{40} $ for {\bf Case 2}.   We thus conclude that it is NP-hard to compute a $(1+\frac{1}{5N})$-approximate   optimal  public signaling scheme.

 		\begin{lemma}\label{lemma.1}
 			If $G$ is from {\bf Case 1}, the optimal public signaling scheme achieves expected social cost at most $N - (1-\epsilon) $ at any Nash equilibrium. 
 		\end{lemma}

 		\begin{lemma}\label{lemma.2}
 			If $G$ is from {\bf Case 2}, the optimal public signaling scheme achieves expected social cost at least $ N  - \frac{3}{4} - \frac{1}{(1-\epsilon)q^k}   $ at any Nash equilibrium. 
 		\end{lemma} 
 		
 		The proof of the above two   lemmas are technical and deferred to Appendix \ref{appendix:2lemmas}. The proof of Lemma \ref{lemma.1} constructs   a good public signaling scheme based on the coloring of a graph from {\bf Case 1}, essentially by revealing the color of the realized state (a node). We can show that   any equilibrium under this public signaling scheme has expected social cost at most $N  - (1-\epsilon) $.

 		Much more involved is the proof of Lemma \ref{lemma.2}. We start by exhibiting multiple important properties about   \emph{any}  Nash equilibrium of the constructed game under \emph{any} posterior distribution $ p \in \Delta_{\Theta}$. Let $S$ denote the set of resources in $V$ with at least one agents at equilibrium.  First, we show that $S$ must be an independent set. Second, we show that: (1) either no agents will go to the backup resource $0$; (2) or $p_{\theta}$  cannot be too large for any $\theta \in S$ in the sense that it will be properly upper bounded by its neighbors' total posterior probabilities   $ \sum_{r \in adj(\theta)} p_r$.  We then argue that in both cases, the social costs cannot be too small.  
 		Intuitively, the former case (1)  will have high  social cost  because the set $S$ has a small size $ N \cdot \frac{1}{(1-\epsilon)q^k}$ in \textbf{Case 2} and all the $N$ agents competing among these few many resources in $S$ will lead to large congestion cost. The argument in this part crucially relies on our carefully chosen quadratic congestion function $1-\frac{1}{n^2}$ which quickly increases   as $n$ becomes large. 
 		%\textcolor{red}{(Specifically, other congestion function formats like $a+bn$ or $a-b/n$ do not seem to work for our reduction.)} 
 		The later case (2) will also have high social cost because the probability  that resource $\theta$ is a good resource, i.e.,  $p_{\theta}$,  is not large. The technical arguments to concretize these intuition turns out to be   involved and are deferred to the appendix. % \thanh{I think more details for Lemma 3 need to be added here.}

 		% Next, we will show that if $G$ is from {\bf Case 2}, the optimal public signaling scheme will induce expected social cost at least $2N^2 - 
 		% \frac{3}{4}N$ with properly chosen parameters $\epsilon, q, k$ (in particular, $ \frac{4}{(1-\epsilon)q^k} \leq 3/4$). 
 		
 		% Armed with Lemma \ref{lem:SignalIndep}, we are now ready to prove the cost lower bound for {\bf Case 2}. The following lemma shows that . Since the hardness of the independent set instance holds for any parameter $\epsilon, q, k$, we can simply choose them to satisfy $\frac{4}{(1-\epsilon)q^k } \leq \frac{3}{4}$, which will lead to $ 2N^2 - \frac{3}{4}N$ as the lower bound of the social cost of the optimal public signaling scheme, and thus conclude our proof.  

 	\end{proof}

 	\subsection{Evidence of intractability for optimal private signaling.}
 	% \todo{Move private signaling formulation here and show that its dual problem has an NP-hard separation oracle, suppose we are allowed to show negative parameters}
 	
 	% \todo{Formulate the optimal correlated equilibrium problem here (note, not the private signaling even, just optimal correlated equilibrium problem. )}
 	
 	% \hf{Okay, so computing optimal CE in singleton congestion game is an open problem, as mentioned in Albert's paper titled \emph{A General Framework for Computing Optimal Correlated Equilibria in
 			% Compact Games} "We have shown that for singleton congestion games, the optimal social welfare problem and the optimal CCE problem are tractable while the complexity of the optimal CE problem is unknown. An open problem is to prove a separation of the complexities of these problems for singleton congestion games or for another class."}

 	% (1) Pei-Ping Shen et al. Global optimization for sum of linear ratios problem with coefficients
 	
 	% (2) Takahito Kuno. A branch-and-bound algorithm for maximizing the sum of several linear ratios
 	
 	% (3) 
 	% John Gunnar Carlsson et al. A linear relaxation algorithm for solving the sum-of-linear-ratios problem with lower dimension
 	
 	% (4) Yong-Hong Zhang et al. A New Branch and Reduce Approach for Solving Generalized Linear Fractional Programming
 	
 	% (5) 
 	% Hong-Wei Jiao et al. A practicable branch and bound algorithm for sum of linear ratios problem
 	
 	Lastly, we move to optimal private signaling. It turns out that understanding the complexity of this problem is challenging even when the problem instance degenerates to a single state of nature, in which case the optimal private signaling degenerates to computing the \emph{optimal correlated equilibrium}. We provide a strong evidence of hardness for this problem by proving that obtaining an efficient separation oracle for its \emph{dual} linear program is NP-hard (Conjecture~\ref{con.1}) even for symmetric \scg{}s with linear latency functions. As mentioned previously, this open question is interesting since both the socially-optimal coarse correlated equilibrium (CCE) and the socially optimal Nash equilibrium admit polynomial times, as shown by   \citet{ Castiglioni20} and  \citet{ ieong2005fast} respectively. Therefore,  the hardness of optimal correlated equilibrium will be surprising and intriguing phenomenon.         
 	
 	\begin{conjecture}\label{con.1}
 		Computing the social-cost-minimizing correlated equilibrium in a symmetric \texttt{SCG} is NP-hard. 
 	\end{conjecture}
 	We remark that for an instance of \scg{} to be hard, it must have truly mixed optimal correlated equilibrium, i.e., randomizing over multiple action profiles. This is because any pure-strategy correlated equilibrium is also a pure-strategy Nash equilibrium, which can be computed in polynomial time in \scg{}s \cite{ ieong2005fast}. We found that this is a key challenge in constructing hard instances. Previous  proof techniques for the hardness of   optimal correlated equilibrium in, e.g., general congestion games, facility location games, network design games etc. \cite{ papadimitriou2008computing}, cannot be easily adapted to \scg{}s since they are all   based on constructing instances in which the optimal pure Nash coincides with the optimal correlated equilibrium and is hard to compute.  
 	
 	Next, we present our evidence for Conjecture \ref{con.1}. We start with the   formulation of LP \eqref{lp:private-naive} for  the optimal private signaling problem and restrict to its degenerated case with a single state. LP \eqref{lp:private-naive} with $|\Theta| = 1$ degenerates to an LP for computing the optimal correlated equilibrium. Algebraic calculation shows that solving this degenerated LP reduces to obtaining a separation oracle for its dual program, which turns out to be the following problem: 
 	% 
 	% \begin{equation}\label{lp:ce}
 		% 	\begin{aligned}
 			% 		\min \quad & \sum_{\bv{a} \in A} \pi(\bv{a}) \sum_{r \in [R]} n_r c_r(n_r) \\
 			% 		\text{s.t.} \quad &\sum_{\bv{a}: a_i=r} \pi(\bv{a}) \left( c_{r}(n_{r}) -  c_{r'}(n_{r'}+1) \right) \leq 0, \forall r, r' \in [R], i \in [N] \\
 			% 		& \sum_{\bv{a} \in A} \pi(\bv{a}) = 1 \\
 			% 		& \pi(\bv{a}) \geq 0, \quad \forall \textbf{a} \in A \\
 			% 	\end{aligned}
 		% \end{equation}
 	% Solving \eqref{lp:ce} boils down to obtaining a separation oracle for its dual program, which is formulated as follows:
 	% 
 	\begin{equation} \label{incentive-deviate-SC}
 		\min_{\bv{n} \in P(A)} \quad \sum_{r \in [R]} n_r c_r(n_r) - \sum_{r \in [R]} \sum_{i: a_i=r} \sum_{r' \neq r}  \bigg[ c_{r'}(n_{r'} + 1) - c_r(n_r) \bigg] \cdot z^i_{r, r'} 
 	\end{equation}
 	
 	The following proposition shows that the optimization problem \eqref{incentive-deviate-SC} is NP-hard in general. Our proof proceeds by first ``smooth'' the program to  a continuous optimization problem and then prove its hardness by reducing from the maximum independent set for 3-regular graphs (which is APX-hard). We then convert the hardness of the smoothed continuous problem to the hardness of Problem  \eqref{incentive-deviate-SC}. A formal proof is  designate  to Appendix \ref{proof:private-hardness}. 
 	
 	\begin{proposition} \label{thm:private-hardness}
 		It is NP-hard to solve Optimization Problem \eqref{incentive-deviate-SC} for general $\{z^i_{r,r'} \}_{i, r \not = r'}$, even when all resources have the same linearly increasing congestion functions $c_r(n) = \frac{n}{N}$ and the \scg{} is symmetric. 
 	\end{proposition}
 	
	% Acknowledgments here
\section*{Acknowledgments.}
The work has been supported by NSF grant CCF-2132506.
 
\bibliographystyle{plainnat}  
\bibliography{refer}

\appendix 
 		\section{Proof of Theorem~\ref{thm:opt-private}.} \label{append:const-private-proof}
 		
 		% \begin{theorem}
 			%     The social-cost-minimizing private signaling scheme can be computed in polynomial time when $R$ is a constant.
 			% \end{theorem}
 		We start by re-stating the natural   exponential-size linear program (LP \eqref{lp:private-naive} in the main context) for computing the optimal private signaling, with variable $\pi(\textbf{a}|\theta)$ as the  probability of recommending resource $a_i$ to agent $i$ conditioned on state of nature $\theta$. 
 		\begin{subequations}
 			\label{const-lp}
 			\begin{align}
 				\min \quad & \sum_{\theta\in \Theta, \textbf{a} \in A} \mu(\theta) \pi(\textbf{a}|\theta) \sum_{r \in [R]}n_r c^{\theta}_{r}(n_r) \label{const-lp-obj} \\
 				\textup{s.t.} \quad & \sum_{\theta \in \Theta} \sum_{\textbf{a} \in A: a_i=r} \mu(\theta)  \pi(\textbf{a}|\theta) \left[ c^{\theta}_{r}(n_{r}) - c^{\theta}_{r'}(n_{r'}+1) \right] \leq 0, \, \forall i \in [N], r, r' \in A_i \label{const-lp-ic}\\
 				& \sum_{\textbf{a} \in A} \pi(\textbf{a}|\theta)=1, \quad \forall \theta \in \Theta \\
 				& \pi(\textbf{a}|\theta)\geq 0, \quad \forall \theta \in \Theta; \textbf{a} \in A
 			\end{align}
 		\end{subequations}
 		
 		Constraint \eqref{const-lp-ic} means that any agent $i$ recommended to resource $r = a_i$ will \emph{not} prefer choosing any other resource $r'$ in her available set $A_i$. The remaining constraints enforce a feasible private signaling scheme whereas the LP objective is to minimize the expected social cost. Note that, in the above formulation, the configuration $n_r(\bv{a})$ for each resource is determined by the corresponding action profile $\bv{a}$, but we omitted $\bv{a}$ for notation convenience.
 		
 		Unfortunately, \eqref{const-lp} has $\Omega(R^N)$ size and thus cannot be solved efficiently. To design an efficient algorithm, the key idea is identify compact yet still expressive marginal probabilities to capture the signaling scheme. Similar idea has been widely employed in auction design where interim allocation rules are used \cite{ border1991implementation,alaei2012bayesian,cai2012algorithmic}. The main challenge, however, is to characterize the feasible region of the ``interim''  signaling rule of any private scheme. 
 		
 		We observe that though LP \eqref{const-lp} has exponentially many variables, the utility of agent $i$ only depends on three factors: the state of nature $\theta$, the resource $a_i$ she is recommended for, and the number of agents $n_r$ choosing resource $r$. Thus, let $\textbf{n} = (n_1,   ..., n_R) \in P(A)$ denote any feasible configuration. We can rewrite LP \eqref{const-lp} with variables $x_{\theta \textbf{n} i r}$ which denote the probability that conditioned on the state of nature $\theta$, agent $i$ is recommended to resource $r$ and the resulting configuration is $\textbf{n}$. We will thus re-write LP \eqref{const-lp} using this new set of variables, as stated in the following lemma. 
 		
 		\begin{lemma} \label{lem:marginalLP}
 			Given any private signaling scheme $\pi$, let variable 
 			\begin{equation} \label{eq:interim-def}
 				x_{\theta \textbf{n} i r} = 
 				\sum_{\textbf{a} \in A} \pi(\textbf{a}|\theta) \cdot \mathbb{I}(a_i = r, \bv{n}(\bv{a}) = \bv{n} )
 			\end{equation} 
 			denote the marginal probability that conditioned on the state of nature $\theta$, agent $i$ is recommended to choose resource $r$ and the resulting configuration is $\textbf{n}$. Then the objective equation \eqref{const-lp-obj} of LP \eqref{const-lp} is equivalent to 
 			\begin{equation} \label{const-lp-p-obj}
 				\begin{aligned}
 					\min \quad & \sum_{\theta \in \Theta} \mu(\theta) \sum_{\textbf{n} \in P(A)} \sum_{i \in [N]} \sum_{r \in [R]: n_r > 0} x_{\theta \textbf{n} i r} \cdot c^{\theta}_{r}(n_{r}) \\
 				\end{aligned}
 			\end{equation}
 			and the obedience constraint \eqref{const-lp-ic} is equivalent to 
 			\begin{equation} \label{const-lp-p-ic}
 				\sum_{\theta \in \Theta} \mu(\theta) \sum_{\textbf{n} \in P(A)} x_{\theta \textbf{n} i r} \left[ c^{\theta}_{r}(n_r) - c^{\theta}_{r'}(n_{r'} + 1) \right] \leq 0, \quad \forall i \in [N], r, r' \in A_i
 			\end{equation}
 		\end{lemma}
 		
 		\begin{proof}
 			This proof follows from standard probability manipulations. We first derive a useful relation. For any state of nature $\theta$, with slight abuse of notation let  $\mathbb{P}(\bv{n}|\theta) = \mathbb{P}(\bv{a}: \bv{n}(\bv{a}) = \bv{n}|\theta)$ denote the probability that the configuration is $\textbf{n}$ under the signaling scheme $\pi$. We have for any fixed $\theta \in \Theta, \bv{n} \in P(A)$, and $r \in [R]$ such that,
 			\begin{equation}\label{eq:feasibility-proof-step1}
 				\begin{aligned}
 					\sum_{i \in [N]} x_{\theta \bv{n} i r} &= \sum_{i \in [N]} \sum_{\bv{a} \in A} \pi(\bv{a}|\theta) \cdot \mathbb{I}(a_i = r, \bv{n}(\bv{a}) = \bv{n} ) \\
 					&= \sum_{\bv{a} \in A} \pi(\bv{a}|\theta) \cdot\bigg[ \sum_{i \in [N]} \mathbb{I}(a_i = r, \bv{n}(\bv{a}) = \bv{n}) \bigg] \\ 
 					&= \sum_{\textbf{a} \in A} \pi(\bv{a}|\theta) \cdot \bigg[ n_r \cdot \mathbb{I}(\bv{n}(\bv{a}) = \bv{n} ) \bigg]   \\ 
 					&= n_r\sum_{\textbf{a} \in A} \pi(\textbf{a}|\theta) \cdot  \mathbb{I}(\bv{n}(\bv{a}) = \bv{n} )    \\
 					&= n_r\mathbb{P}(\textbf{n}|\theta) \forall r \in [R]
 				\end{aligned}
 			\end{equation}
 			% \thanh{Equation (15) is only true with $n_r > 0$; or you can rewrite $\sum_{i \in [N]} x_{\theta \bv{n} i r} = n_r \mathbb{P}(\textbf{n}|\theta)$ to avoid the zero-dividing}
 			Then, the objective of linear program, \eqref{const-lp-obj}, can be written as
 			% \begin{equation*}
 				% \begin{aligned}
 					%     & \sum_{\theta \in \Theta} \mu(\theta) \sum_{\bv{a} \in A} \pi(\bv{a}|\theta) \sum_{r \in [R]}  n_{r} c^{\theta}_{r}(n_{r}) \\
 					%     =& \sum_{\theta \in \Theta} \mu(\theta) \sum_{\bv{n} \in P(A)} \sum_{\bv{a}: \bv{n}(\bv{a}) = \bv{n}} \pi(\bv{a}|\theta) \sum_{r \in [R]}  n_{r} c^{\theta}_{r}(n_{r}) \\
 					%     =& \sum_{\theta \in \Theta} \mu(\theta) \sum_{\bv{n} \in P(A)} \mathbb{P}(\bv{n}(\bv{a})=\bv{n}|\theta) \sum_{r \in [R]} n_r c^{\theta}_{r}(n_{r}) \\
 					%     =& \sum_{\theta \in \Theta} \mu(\theta) \sum_{\bv{n} \in P(A)} \frac{1}{n_r} \sum_{i \in [N]} x_{\theta \bv{n} i r} \sum_{r \in [R]} n_{r} c^{\theta}_{r}(n_{r}) \\
 					%     =& \sum_{\theta \in \Theta} \mu(\theta) \sum_{\bv{n} \in P(A)} \sum_{i \in [N]} \sum_{r \in [R]} x_{\theta \bv{n} i r}  c^{\theta}_{r}(n_{r})
 					% \end{aligned}
 				% \end{equation*}
 			
 			\begin{equation*}
 				\begin{aligned}
 					& \sum_{\theta \in \Theta} \mu(\theta) \sum_{\bv{a} \in A} \pi(\bv{a}|\theta) \sum_{r \in [R]}  n_{r} c^{\theta}_{r}(n_{r}) \\
 					% & \sum_{\theta \in \Theta} \mu(\theta) \sum_{\bv{a} \in A} \pi(\bv{a}|\theta) \sum_{r \in [R]: n_r > 0}  n_{r} c^{\theta}_{r}(n_{r}) \\
 					=& \sum_{\theta \in \Theta} \mu(\theta) \sum_{\bv{n} \in P(A)} \sum_{\bv{a}: \bv{n}(\bv{a}) = \bv{n}} \pi(\bv{a}|\theta) \sum_{r \in [R]}  n_{r} c^{\theta}_{r}(n_{r}) \\
 					=& \sum_{\theta \in \Theta} \mu(\theta) \sum_{\bv{n} \in P(A)} \mathbb{P}(\bv{n}|\theta) \sum_{r \in [R]} n_r c^{\theta}_{r}(n_{r}) \\
 					=& \sum_{\theta \in \Theta} \mu(\theta) \sum_{\bv{n} \in P(A)}  \sum_{r \in [R]} \left[n_r\mathbb{P}(\bv{n}|\theta)\right] c^{\theta}_{r}(n_{r}) \\
 					=& \sum_{\theta \in \Theta} \mu(\theta) \sum_{\bv{n} \in P(A)}  \sum_{r \in [R]}\sum_{i \in [N]} x_{\theta \bv{n} i r}  c^{\theta}_{r}(n_{r})
 				\end{aligned}
 			\end{equation*}
 			
 			Next, we write the persuasive constraint for each agent $i \in [N]$, each recommended resource $r \in A_i$ and each alternative resource $r' \in A_i$.
 			\begin{equation*}
 				\begin{aligned}
 					& \sum_{\theta \in \Theta} \mu(\theta)\sum_{\bv{a} \in A: a_i=r} \pi(\bv{a}|\theta) \left[ c^{\theta}_{r}(n_{r}) - c^{\theta}_{r'}(n_{r'}+1) \right] \\
 					=& \sum_{\theta \in \Theta} \mu(\theta) \sum_{\bv{n} \in P(A)} \sum_{\bv{n}(\bv{a})=\bv{n}; a_i = r} \pi(\bv{a}|\theta) \left[ c^{\theta}_{r}(n_{r}) - c^{\theta}_{r'}(n_{r'} + 1) \right] \\
 					=& \sum_{\theta \in \Theta} \mu(\theta) \sum_{\textbf{n} \in P(A)} x_{\theta \textbf{n} i r} \left[ c^{\theta}_{r}(n_{r}) - c^{\theta}_{r'}(n_{r'}+1) \right]
 				\end{aligned}
 			\end{equation*}
 			
 			Thus, the original inequality is the same as 
 			$$\sum_{\theta \in \Theta} \mu(\theta) \sum_{\textbf{n} \in P(A)} x_{\theta \textbf{n} i r} \left[ c^{\theta}_{r}(n_{r}) - c^{\theta}_{r'}(n_{r'}+1) \right] \leq 0$$
 			which concludes our proof.
 		\end{proof}
 		
 		Next, we will show a key Lemma to the proof of Theorem \ref{thm:opt-private} that compactly characterizes the feasible marginal probabilities introduced in Equation \eqref{eq:interim-def}. 
 		
 		% \thanh{I rewrote the statement of Lemma 5 a bit to address the issue when $n_r = 0$}
 		\begin{lemma} \label{lem:lp-feasibility}
 			The set of variables $\{x_{\theta \textbf{n} i r} \geq 0: \theta \in \Theta, i \in [N], \textbf{n} \in P(A), r \in [R]\}$ satisfies the following constraints
 			\begin{subequations} \label{const:lp-feasibility}
 				\begin{align} 
 					x_{\theta \bv{n} i r} = 0, \quad & \forall i \in [N], r \notin A_i, \bv{n} \in P(A), \theta \in \Theta \label{const:lp-feasibility-1} \\
 					% \sum_{k=1}^{N} \sum_{\textbf{n} \in P(A): n_r = k} \frac{1}{k} \sum_{i \in [N]} x_{\theta \textbf{n} i r} \leq 1, \quad & \forall r \in [R], \theta \in \Theta  \label{const:lp-feasibility-2} \\
 					% n_r\sum_{r' \in [R]} x_{\theta \textbf{n} i r'} = \sum_{j \in [N]}  x_{\theta \textbf{n} j r}, \quad & \forall i \in [N],  r \in [R], \theta \in \Theta, \textbf{n} \in P(A) \label{const:lp-feasibility-3}
 					\sum_{\textbf{n} \in P(A)} \sum_{r \in [R]} x_{\theta \textbf{n} i r} = 1, \quad & \forall i \in [N], \theta \in \Theta  \label{const:lp-feasibility-2} \\
 					\sum_{j \in [N]}  x_{\theta \textbf{n} j r} = n_r\sum_{r' \in [R]} x_{\theta \textbf{n} i r'} , \quad & \forall i \in [N],  r \in [R], \theta \in \Theta, \textbf{n} \in P(A)  \label{const:lp-feasibility-3}
 				\end{align}
 			\end{subequations}
 			\emph{if and only if} there exists a feasible signaling scheme $\pi$ that induces $X$ as in Equation \eqref{eq:interim-def}. Moreover, given any $X$ satisfies Linear System   \eqref{const:lp-feasibility}, such a $\pi$ that induces $X$ can be found in polynomial time. 
 		\end{lemma}
 		% \chenghan{
 			%     \begin{subequations} 
 				%     \begin{align} 
 					%         \sum_{\textbf{n} \in P(A)} \sum_{r \in [R]} x_{\theta \textbf{n} i r} = 1, \quad & \forall i \in [N], \theta \in \Theta \\
 					%         n_r\sum_{r' \in [R]} x_{\theta \textbf{n} i r'} = \sum_{j \in [N]}  x_{\theta \textbf{n} j r}, \quad & \forall i \in [N],  r \in [R], \theta \in \Theta, \textbf{n} \in P(A) 
 					%     \end{align}
 				%     \end{subequations}
 			% }
 		
 		\begin{proof}
 			We start by proving the ``only if'' direction. That is, suppose $\pi$ is a feasible private signaling scheme and $x_{\theta \textbf{n} i r}$ is defined in Equation \eqref{eq:interim-def}, then all equalities in \eqref{const:lp-feasibility} are satisfied. Obviously, equation \eqref{const:lp-feasibility-1} is satisfied since a feasible signaling scheme $\pi$ will not recommend resources $r \notin A_i$ to agent $i$. 
 			
 			In addition, regarding constraint~(\ref{const:lp-feasibility-2}), $\sum_{r \in [R]} x_{\theta \textbf{n} i r}$ is essentially the probability that the configuration is $\mathbf{n}$ and agent $i$ is sent to a resource given $\theta$. This probability is thus equal to the probability that the configuration is $\mathbf{n}$ given $\theta$. This is because for any configuration, every agent $i$ is always recommended to one of the resources (i.e., $\sum_n n_r = N$). Mathematically, we have:
 			\begin{align*}
 				\sum_{r' \in [R]} x_{\theta \bv{n} i r'} &=\sum_{r' \in [R]}  \sum_{\bv{a} \in A} \pi(\bv{a}|\theta) \cdot  \mathbb{I}(a_i = r', \bv{n}(\bv{a}) = \bv{n} )  \\
 				&=  \sum_{\bv{a} \in A} \pi(\bv{a}|\theta) \cdot \bigg[  \sum_{r' \in [R]}  \mathbb{I}(a_i = r',  \bv{n}(\bv{a}) = \bv{n} ) \bigg] \\
 				&= \sum_{\bv{a} \in A} \pi(\bv{a}|\theta) \cdot \mathbb{I}( \bv{n}(\bv{a}) = \bv{n} ) \\
 				&=\mathbb{P}(\bv{n}|\theta)
 			\end{align*}
 			
 			Therefore, we obtain the following equality (i.e., constraint~(\ref{const:lp-feasibility-2})):
 			\begin{align*}
 				& \sum_{\mathbf{n}\in P(A)} \sum_{r \in [R]} x_{\theta \textbf{n} i r} = \sum_{\mathbf{n}\in P(A)} \mathbb{P}(\mathbf{n}\mid \theta) = 1
 			\end{align*}
 			
 			Finally, by leveraging the proof of Lemma~\ref{lem:marginalLP}, we obtain constraint~(\ref{const:lp-feasibility-3}): 
 			\begin{align*}
 				& \sum_{j \in [N]}  x_{\theta \textbf{n} j r} = n_r \mathbb{P}(\mathbf{n}\mid\theta) = n_r\sum_{r' \in [R]} x_{\theta \textbf{n} i r'}
 			\end{align*}
 			
 			% As we have shown in Lemma \ref{lem:marginalLP}, $\frac{1}{n_r} \sum_{i \in [N]} x_{\theta \bv{n} i r} = \mathbb{P}(\bv{n} | \theta)$. This equation can help us to verify the validity of equalities in \eqref{const:lp-feasibility}. For Equation \eqref{const:lp-feasibility-2}, we have
 			% % \begin{equation*}
 				% %     \sum_{\bv{n} \in P(A)} \frac{1}{n_r} \sum_{i \in [N]} x_{\theta \bv{n} i r}  = \sum_{\bv{n} \in P(A)} \mathbb{P}(\bv{n}|\theta) = 1.
 				% % \end{equation*}
 			
 			% \begin{equation*}
 				%     \sum_{k=1}^{N} \sum_{\textbf{n} \in P(A): n_r = k} \frac{1}{k} \sum_{i \in [N]} x_{\theta \textbf{n} i r} = \sum_{\bv{n} \in P(A): n_r > 0} \mathbb{P}(\bv{n}|\theta) \leq 1
 				% \end{equation*}
 			
 			% Equation \eqref{const:lp-feasibility-3} holds due to the following derivations:
 			% \begin{eqnarray*}
 				%     \sum_{r' \in [R]} x_{\theta \bv{n} i r'} &=& \sum_{r' \in [R]}  \sum_{\bv{a} \in A} \pi(\bv{a}|\theta) \cdot  \mathbb{I}(a_i = r', \bv{n}(\bv{a}) = \bv{n} )  \\
 				%     &=&   \sum_{\bv{a} \in A} \pi(\bv{a}|\theta) \cdot \bigg[  \sum_{r' \in [R]}  \mathbb{I}(a_i = r',  \bv{n}(\bv{a}) = \bv{n} ) \bigg] \\
 				%     &=& \sum_{\bv{a} \in A} \pi(\bv{a}|\theta) \cdot \mathbb{I}( \bv{n}(\bv{a}) = \bv{n} ) \\
 				%     &=& \mathbb{P}(\bv{n}|\theta) \\
 				%     \implies n_r \sum_{r' \in [R]} x_{\theta \bv{n} i r'}&=&n_r \mathbb{P}(\bv{n}|\theta)=  \sum_{i \in [N]} x_{\theta \bv{n} i r}
 				% \end{eqnarray*}

 			Next, we prove the ``if'' direction. That is, given any $X = \{ x_{\theta \textbf{n} i r}: \theta \in \Theta, \bv{n} \in P(A), i \in [N], r \in [R] \}$ that is feasible to Equation \eqref{const:lp-feasibility}, we will find in polynomial time a signaling scheme $\pi$ that induces the given $X$. Essentially, we implement a sampling procedure, which is divided into two steps: (1) \textbf{Step 1} --- sample the configuration $\bv{n}$ for any $ \bv{n}$ conditioned on state $\theta$; (2) \textbf{Step 2} --- sample $\bv{a} \in \{\bv{a}: \bv{n}(\bv{a}) = \bv{n} \}$, i.e., the set of all possible action profiles with the same configuration $\bv{n}$. \textbf{Step 1} is easy. We simply sample $\bv{n}$ with probability $\mathbb{P}(\bv{n}|\theta)=\sum_{r' \in [R]} x_{\theta \bv{n} i r'}$. This step is feasible given constraint~(\ref{const:lp-feasibility-2}). 
 			
 			More involved is  \textbf{Step 2}, which we now describe in detail. \emph{Conditioned} on given state $\theta$ and that $\bv{n}$ is sampled, the marginal probability that agent $i$ is recommended with resource $r$ is required to be:
 			\begin{equation}\label{eq:feasibility-proof-step2-1}
 				\PP(i \to r| \bv{n}, \theta) =  \frac{x_{\theta \bv{n} i r}}{ \sum_{r' \in [R]} x_{\theta \bv{n} i r'}},
 			\end{equation}
 			which satisfies $\sum\limits_{r \in [R]} \PP(i \to r| \bv{n}, \theta) = 1$ for any agent $i \in [N]$. Moreover, for any fixed resource $r$, we have: 
 			\begin{equation} \label{eq:feasibility-proof-step2-2}
 				\sum_{i\in [N]} \PP(i \to r| \bv{n}, \theta) = \frac{ \sum_{i \in [N]}  x_{\theta \bv{n} i r}}{ \sum_{r' \in [R]} x_{\theta \bv{n} i r'}}  =  \frac{ n_r \sum_{r' \in [R]} x_{\theta \bv{n} i r'} }{ \sum_{r' \in [R]} x_{\theta \bv{n} i r'}} = n_r,\\ 
 			\end{equation}
 			where the middle equation follows from Equation \eqref{const:lp-feasibility-3}.  Therefore, given $X$, we can view the probabilities $ \PP(i \to r| \bv{n}, \theta)$ defined in Equation \eqref{eq:feasibility-proof-step2-1} as a \emph{fractional flow} in a bipartite graph with left-side nodes as agents and right-side nodes as resources. The flow amount from agent $i$ to resource $r$ is $ f(i \to r) = \PP(i \to r| \bv{n}, \theta)$; the total amount of flow going outside from agent $i$ is $1$ and the total amount of flow entering resource $r$ is $n_r$ as required by Equation \eqref{eq:feasibility-proof-step2-2}.
 			Note that any integer flow satisfying these supply-demand constraints (i.e., $1$ unit of flow is supplied from any agent node $i$ and $n_r$ units of flow are demanded at resource $r$ node) corresponds precisely to a deterministic action profile $\bv{a}$ satisfying precisely $\bv{n}(\bv{a}) = \bv{n}$. If we can decompose the fractional flow into a distribution over such integer flows, which are action profiles in $\{\bv{a}: \bv{n}(\bv{a}) = \bv{n}\}$, this decomposition naturally corresponds to (part of) a private signaling scheme $\pi(\bv{a}; \theta)$ for those $\bv{a} \in \{\bv{a}: \bv{n}(\bv{a}) = \bv{n}\}$ that satisfies Equation \eqref{const:lp-feasibility-3}. Together with \textbf{Step 1} of sampling the configuration $\mathbf{n}$, we obtain private signaling scheme with size $O(N^R N R)$ that matches any given $X$ to system \eqref{const:lp-feasibility}. The full algorithm combining both steps is described in Algorithm \ref{alg:const-sample-algo}.

 			% \begin{alg}
 				\begin{algorithm}\begin{onehalfspace}\begin{algorithmic}[1]
 							\INPUT Any $X = \{ x_{\theta \textbf{n} i r}: i \in [N], r \in [R], \bv{n} \in P(A), \theta \in \Theta \}$ feasible to Equation \eqref{const:lp-feasibility}
 							\OUTPUT Private signaling scheme $\pi(\bv{a}|\theta)$ that satisfies Equation \eqref{eq:interim-def} with the input $X$ 
 							\STATE Initialize $\pi(\bv{a}|\theta) = 0, \forall \bv{a}\in A, \theta \in \Theta$ 
 							\FOR{ $ \theta \in \Theta, \bv{n} \in P(A) $ }
 							\STATE Set $\PP(\bv{n}| \theta) = \sum_{r' \in [R]} x_{\theta \bv{n} i r'}$
 							\FOR{$ i \in [N], r \in [R] $ } 
 							\STATE Compute  $\PP(i \to r| \bv{n}, \theta) =  \frac{x_{\theta \bv{n} i r}}{ \sum_{r' \in [R]} x_{\theta \bv{n} i r'}}$
 							\ENDFOR 
 							\STATE Construct a fully connected bipartite graph $G = ([N] \cup [R], E)$
 							\STATE Construct flow $f$ on $G$ with flow amount $f(i \to r) = \PP(i \to r| \bv{n}, \theta)$ on edge $(i, r) \in E$
 							\STATE \label{alg:step-decompose}Decompose fractional flow $f$ on $G$ into a distribution over integer flows $(f_1, p_1),\cdots (f_K, p_K)$ where $f_k$ is an integer flow and $p_k$ is its probability
 							\STATE For $k=1, \cdots, K$, let $\bv{a}$ be the action profile corresponding to $f_k$ and set $\pi(\bv{a}|\theta) = p_k \PP(\bv{n}| \theta)$
 							\ENDFOR 
 							\caption{\small Sampling a private signaling scheme $\pi(\bv{a}|\theta)$ from any  $X$ feasible to System \eqref{const:lp-feasibility} }
 							\label{alg:const-sample-algo}
 				\end{algorithmic}\end{onehalfspace}\end{algorithm}
 				% \end{alg}
 			
 			The only missing part now is an algorithm that decomposes the fractional flow into a distribution over integer flow satisfying the same demand and supply at all nodes (i.e., Step \ref{alg:step-decompose} in Algorithm \ref{alg:const-sample-algo}). This can be done by converting the flow into a fractional max-flow in a modified graph $G'$ with: (1) a super source $s$ connected to all agent nodes $i$ with capacity $1$ and also flow amount $1$; (2) all resource nodes $r$ connecting to a super sink node $t$ with capacity $n_r$ and flow amount $n_r$ as well. A visual description can be found in Figure \ref{fig:flow-interpret}. It is easy to see that $f$ constructed in Step \ref{alg:step-decompose}  is a max-flow in this graph. This max fractional flow can be decomposed as a distribution over integer max flows as follows. 
 			
 			\begin{figure}[t!]
 				\begin{center}
 					\scalebox{.9}{\begin{tikzpicture}[
    mycircle/.style={
        circle,
        draw=black,
        fill=white,
        fill opacity = 0.3,
        text opacity=1,
        inner sep=0pt,
        minimum size=15pt,
        font=\large},
    cdotscircle/.style={
        fill=white,
        text opacity=1,
        inner sep=0pt,
        minimum size=15pt,
        font=\large},
      smalltextcircle/.style={
        fill=white,
        text opacity=1,
        inner sep=0pt,
        minimum size=1pt,
        font=\small},
      myarrow/.style={-Stealth},
      node distance=0.6cm and 1.2cm
      ]
    
    \node[smalltextcircle] (agent) at (3,3.5) {Agent node $i$};
    \node[smalltextcircle] (agent) at (7,3.9) {Resource};
    \node[smalltextcircle] (agent) at (7,3.5) {node $r$};
    
    \node[smalltextcircle] (agent) at (0.8,-0.15) {$(s, i)$:};
    \node[smalltextcircle] (agent) at (2.1,0) {capacity: $1$};
    \node[smalltextcircle] (agent) at (1.8,-0.3) {flow: $1$};

    \node[smalltextcircle] (agent) at (3.9,-3.4) {$(i, r)$:};
    \node[smalltextcircle] (agent) at (5.2,-3.3) {capacity: $1$};
    \node[smalltextcircle] (agent) at (5.45,-3.65) {flow: $f(i \rightarrow r)$};
    
    \node[smalltextcircle] (agent) at (9,-1.6) {$(r, t)$:};
    \node[smalltextcircle] (agent) at (10.4,-1.5) {capacity: $n_r$};
    \node[smalltextcircle] (agent) at (10.1,-1.8) {flow: $n_r$};
    
    \node[mycircle] (s) at (0,0) {$s$};
    
    \node[mycircle] (i1) at (3,3) {$1$};
    \node[mycircle] (i2) at (3,1) {$2$};
    \node[cdotscircle] (icdots) at (3,-1) {$\cdots$};
    \node[mycircle] (iN) at (3,-3) {$N$};

    \draw[myarrow] (s) -- (i1) node[pos=.5,above] {};
    \draw[myarrow] (s) -- (i2) node[pos=.5,above] {};
    \draw[myarrow] (s) -- (iN) node[pos=.5,above] {};
    
    \node[mycircle] (l0) at (7,3) {$1$};
    \node[mycircle] (l1) at (7,1) {$2$};
    \node[cdotscircle] (lcdots) at (7,-1) {$\cdots$};
    \node[mycircle] (lL) at (7,-3) {$R$};
    
    \draw[myarrow] (i1) -- (l0) node[pos=.5,above] {};
    \draw[myarrow] (i1) -- (l1) node[pos=.5,above] {};
    \draw[myarrow] (i1) -- (lL) node[pos=.5,above] {};
    
    \draw[myarrow] (i2) -- (l0) node[pos=.5,above] {};
    \draw[myarrow] (i2) -- (l1) node[pos=.5,above] {};
    \draw[myarrow] (i2) -- (lL) node[pos=.5,above] {};
    
    \draw[myarrow] (iN) -- (l0) node[pos=.5,above] {};
    \draw[myarrow] (iN) -- (l1) node[pos=.5,above] {};
    \draw[myarrow] (iN) -- (lL) node[pos=.5,above] {};

    \node[mycircle] (t) at (10,0) {$t$};
    \draw[myarrow] (l0) -- (t) node[pos=.5,above] {};
    \draw[myarrow] (l1) -- (t) node[pos=.5,above] {};
    \draw[myarrow] (lL) -- (t) node[pos=.5,above] {};
\end{tikzpicture}}
 					\caption{Interpreting marginal probabilities as flow. }
 					\label{fig:flow-interpret}
 				\end{center}
 			\end{figure}
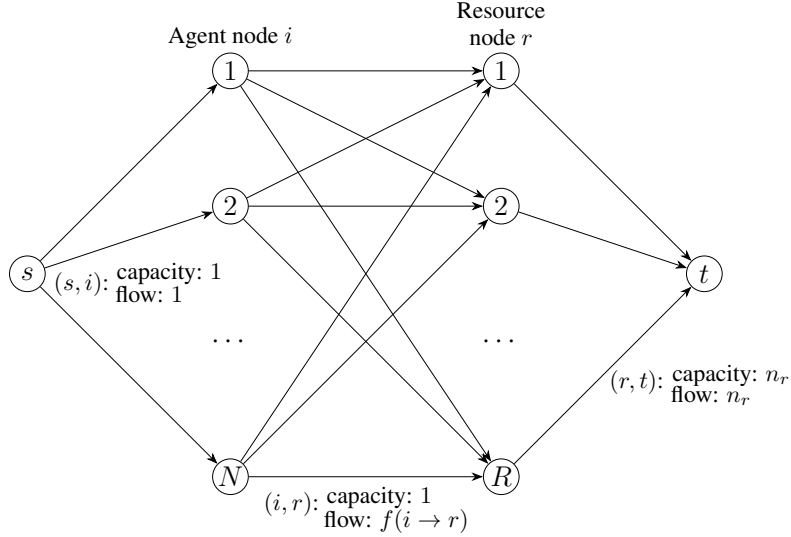
 			
 			We run Ford-Fulkerson algorithm on the edges with strictly positive flow. Since the max flow of this subgraph is $N$, Ford-Fulkerson will be able to find an integer max flow $f_1$. Let $E(f_1)$ denote all the edges that $f_1$ uses and $p_1 = \min_{e} \frac{f(e)}{f_1(e)}$.  For any edge  $e \in E(f_1)$, re-setting their flows to be $f(e) \leftarrow (f(e) - p_1 f_1(e))$, which is non-negative by our choice of $p_1$. Moreover, the edge $e^* \in \arg \min_{e} \frac{f(e)}{f_1(e)}$ will have flow $0$ now. Note that the new flow still satisfies flow conservation constraints. In particular, any edge $(s,i)$ from source $s$ to agent node $i$ must have flow amount $(1- p_1)$ since the max integer flow must flow $1$ unit flow through all these edges.  Similarly, any edge $(r,t)$ from resource node $r$ to sink $t$ must have flow amount $(1- p_1) n_r$ since the max integer flow must flow $n_r$ unit flow through all these edges. We claim that the max-flow in the subgraphs with the new flow amount is still $N$ since dividing the current flow by $(1-p_1)$ gives rise to another fractional $N$-unit flow. We thus can continue this decomposition. This procedure will end in at most $N \times R + N + R$ rounds since each iteration deletes at least one edge. When it terminates, we obtain a distribution over integer flows and each integer flow corresponds to precisely an action profile determined by the selected edges $(i, r)$. 
 		\end{proof}
 		
 		The optimal private signaling scheme can then by computed by solving LP in Figure~\ref{fig:lp-const-private} with variables $x_{\theta \textbf{n} i r}$ and $O(|\Theta| N^{R} R)$ size, and then sample a private signaling scheme  according to Lemma \ref{lem:lp-feasibility}. Therefore, the  social-cost-minimizing private signaling scheme can be computed in polynomial time when $R$ is any constant.

 		% \subsection*{Corollary Following from Theorem \ref{thm:opt-private}}
 		
 		% In addition, the following corollary about polynomial-time solvability of optimal ex-ante private signaling scheme follows easily from Theorem \ref{thm:opt-private}. 
 		
 		% \begin{corollary}\label{cor:ex-ante-const}
 			% The following variant of LP computes the social-cost-minimizing ex-ante private signaling scheme in polynomial time when $R$ is a constant. 
 			% \begin{align}
 				% \label{const-ex-ante-lp}
 				%     \min \quad & \sum_{\theta \in \Theta} \mu(\theta) \sum_{\bv{n} \in P(A)} \sum_{i \in [N]} \sum_{r \in [R]} x_{\theta \bv{n} i r} c^{\theta}_{r}(n_{r}) \\
 				%     \textup{s.t.} \quad & \sum_{\theta \in \Theta, r \in [R]} \sum_{\bv{n} \in P(A)} \mu(\theta)  x_{\theta \bv{n} i r} c^{\theta}_{r}(n_r) \leq \\ 
 				%     & \qquad \sum_{\theta \in \Theta, r \in [R]} \sum_{\bv{n} \in P(A)} \mu(\theta) x_{\theta \bv{n} i r} c^{\theta}_{r'}(n_{r'} + \mathbb{I}(r \not = r')), \, \forall i \in [N], r' \in A_i\\
 				% & \text{Constraints } (\ref{const:lp-feasibility-1}-\ref{const:lp-feasibility-3}) \text{ are satisfied} 
 				% \end{align}
 			% \end{corollary}

 		% \section{Omitted Proofs in Section \ref{sec:public-many}}\label{append:public-many}
 		% \subsection{Proof of Lemma~\ref{lemma.1}} \label{appendix:proof-public-hardness-case1}

 		% \subsection{Proof of Lemma~\ref{lemma.2}} \label{appendix:proof-public-hardness-case2}

 		\section{Proof of technical Lemma \ref{lemma.1} and \ref{lemma.2} }\label{appendix:2lemmas}
 		\subsection{Proof of   Lemma   \ref{lemma.1} }
 		
 		We construct such a public signaling scheme based on the coloring of a graph from {\bf Case 1}. Intuitively, we will try to pool the good resource with bad resources so make it as tractable as any normal resource.  Specifically, since $G$ is from {\bf Case 1}, let us fix a coloring of $(1-\epsilon)R$ vertices/resources with colors $k=1, \ldots, q$, and we use this coloring to construct a public scheme achieving expected social cost at most $N  - (1-\epsilon) $. The scheme uses $q+1$ public signals, and is defined as follows: if $\theta$ has color $k$ then deterministically send the signal $k$, and if $\theta$ is uncolored then deterministically send the signal $0$. Given any public signal $k > 0$, the posterior distribution on states of nature is the uniform distribution over the resources with color $k$, which forms an independent set $S_k$ of size $\frac{1-\epsilon}{q} R$ and equals precisely the number of agents $N$ by our construction. It is easy to verify that any resource $r \in S_k$ is a good resource with probability $1/N$ and is a normal resource with probability $(N-1)/N$. Notably it is never a bad resource because $r$ is not adjacent to other possibly good resources in $S_k$ since $S_k$ is an independent set. On the other hand, any resource $r \not \in S_k$ has no chance to be a good resource, and can be either a normal or a bad resource. Consequently, the unique equilibrium under this public signaling scheme is that each agent picks one resource from $S_k$, which leads to expected social cost $N \times [(1  - \frac{1}{N} ) \times 1 + \frac{1}{N} \times 0 ] = N -1 $ (the `$0$'' is the social cost of any good resource with $n=1$ agent on it). 
 		
 		Finally, the public signal $k = 0$ will be sent with probability $\epsilon$. The cost for this public signal is at most $1 $ since the agents can at least choose resource $0$. Overall, this shows that there exists a public signal that achieves social cost at most $N -  1+ \epsilon $.
 		
 		\subsection{Proof of   Lemma   \ref{lemma.2}. }
 		To prove the social cost lower bound in {\bf Case 2}, we start with a lemma  which characterizes an important property of any public signaling scheme under any Nash equilibrium. 
 		\begin{lemma} \label{lem:SignalIndep}
 			Let   $ p \in \Delta_{\Theta}$ be the posterior distribution of any public signal. Consider \emph{any}  Nash equilibrium under public signal $x$, characterized by a configuration $\bvec{n}$. Let $S = \{ r \subseteq V: n_{r}>0 \} $ denote the set of resources in $V$ with at least one agent under the NE. Then we have 
 			\begin{enumerate}
 				\item Set $S$ is an independent set of $G$.
 				\item Either \quad $\sum_{r \in S} n_{r} = N$; 
 				\\ 
 				Or \qquad \,  for any $r \in S$, $p_{r} \cdot   \frac{1}{(n_{r}+1)^2} - 2 \Big[ \sum_{r' \in adj(r)} p_{r'} \Big] \leq 0$.
 			\end{enumerate}
 		\end{lemma}
 		
 		\begin{proof}
 			% \thanh{Even though we have resources and states correspond to the same vertex set, the notations for resources and for states should be consistent. For example, $r$ for resource and $\theta$ for state}
 			To prove the first claim, we show that for any $r \in S$ we must have $p_{r} >  \sum_{r' \in adj(r)} p_{r'} $ where $adj(r) \subset V$ is the set of all adjacent nodes of $r$ in $G$. This implies that the set $S$ must form an independent set since any two adjacent resources $r, \hat{r} \in S$ cannot simultaneously satisfy $p_{r} >  \sum_{r' \in adj(r)} p_{r'}  \geq p_{\hat{r}}$ and $p_{\hat{r}} >  \sum_{r' \in adj(\hat{r})} p_{r'}  \geq p_{r}$. 
 			
 			To show $p_{r} >  \sum_{r' \in adj(r)} p_{r'} $, note that any state $r' \in adj(r)$ will make $r$ a bad resource with probability $p_{r'}$ whereas $r$ only has probability $p_{r}$ to be a good resource. We derive the expected congestion at resource $r$, denoted as $C^{p}_{r}(n)$, as follows
 			
 			\begin{eqnarray} \nonumber 
 				C^{p}_{r}(n) &=& p_{r} c_{r}^{r}(n) +  \sum_{r' \in adj(r)} p_{r'} c_{r}^{r'}(n) + \left( 1 - p_{r} - \sum_{r' \in adj(r)} p_{r'} \right)\times 1   \\ \nonumber 
 				&=&  p_{r} \bigg[ 1 - \frac{1}{n^2} \bigg] + \sum_{r' \in adj(r)} p_{r'} \times 3 + \left( 1 - p_{r} - \sum_{r' \in adj(r)} p_{r'} \right)   \\  \label{eq:demand-cost}
 				&=& \Big[ 1 + 2 \sum_{r' \in adj(r)} p_{r'} \Big]   - p_{r} \cdot \frac{1}{n^2} \\ \nonumber  
 			\end{eqnarray}
 			
 			For $r$ to have at least one agent, its expected congestion cost at $n=1$, i.e., $C^{p}_{r}(1) = \Big[ 1 + 2 \sum_{r' \in adj(r)} p_{r'} \Big]   - p_{r}  $ must be strictly less $1$.  Without loss of generality, suppose agents always prefer the safe resource $0$ if there is a tie with resource $0$, i.e., when there is another resource which has expected cost $1$.  since choosing the backup resource $0$ will provide a  safe  cost $1$ regardless of the state. Thus, one necessary condition for the above inequality to hold is $p_{r} >  2 \sum_{r' \in adj(r)} p_{r'} \geq \sum_{r' \in adj(r)} p_{r'}   $. 
 			
 			To prove the second claim, suppose $\sum_{r \in S} n_{r} = N$ does not hold. In other words, we have $\sum_{r \in S} n_{r} < N$. This  implies that there are $N - \sum_{r \in S} n_{r} >0$ agents choosing the backup resource $0$, with congestion cost $1$. The fact that these agents do not choose any other resource $r \in S$ must because that will lead to a higher congestion cost after increasing the number of agents $n_{r}$ by $1$. In other words, this implies the desired inequality:
 			$$ 1 \leq C_{r}^p (n_{r} + 1) = \Big[ 1 + 2 \sum_{r' \in adj(r)} p_{r'} \Big]   - p_{r} \cdot \frac{1}{(n_{r} + 1)^2},$$
 			which  concludes our proof for Lemma~\ref{lem:SignalIndep}.
 		\end{proof}
 		
 		Given any public signal, let distribution $p \in \Delta_{\Theta}$ be its posterior distribution. Due to symmetry, any Nash equilibrium (NE) is completely specified by a configuration $\bvec{n}$. Let set $S = \{r \in V: n_{r} > 0 \}$ denote the set of resources in $V$ with at least one agent. By Lemma \ref{lem:SignalIndep}, we know that $S$ forms an independent set of $G$. Note that $|S| \leq \frac{R}{q^{k+1}} = N\cdot \frac{1}{(1-\epsilon)q^k}$ in {\bf Case 2}.    We now distinguish the two situations characterized by the bullet point 2 of Lemma \ref{lem:SignalIndep}. In the first situation, we show the social cost is lower bounded by $N - \frac{1}{(1-\epsilon)q^k}  - \frac{3}{4}  $ and in the second situation, the social cost is lower bounded by $N  - \frac{3}{4} $. These together yields that the social cost at any equilibrium is at least $N - \frac{1}{(1-\epsilon)q^k}  - \frac{3}{4}  $ as desired. 
 		
 		\subsubsection*{Situation 1: $\sum_{r \in S} n_{r} = N$.}
 		
 		Intuitively, the reason that this case has high social cost is because the set $S$ has a small size $ N \cdot \frac{1}{(1-\epsilon)q^k}$ and all the $N$ agents competing among these few many resources in $S$ will lead to large congestion cost. % Specifically, we will show that the expected congestion cost $C^p_r(n_r)$ for any agent choosing resource $r$ satisfies $C^p_r(n_r) \geq 1 - \frac{6}{(1-\epsilon)Nq^k}$ for any resource $r$. This implies that the agents' total cost, i.e., the social cost, is at least  $ N -  \frac{6}{(1-\epsilon)q^k} $ as desired. 
 		Formally, we first argue that it cannot be the case that $C^p_r(n_r) < 1 - \frac{1}{(1-\epsilon)Nq^k}$ for all $r \in S$ where $S = \{ r \in V: n_{r}>0 \} $ denote the set of resources in $V$ with at least one agent under any NE. This is because $1 - \frac{1}{(1-\epsilon)Nq^k} > C^p_r(n_r) \geq  1 - p_r \cdot \frac{1}{n_r^2}$ implies $n_r < \sqrt{p_r N (1-\epsilon)q^k}$. Therefore, we have 
 		\begin{align*}
 			N &= \sum_{r \in S} n_r  & \\
 			&< \sum_{r \in S} \sqrt{p_r N (1-\epsilon) q^k} & \mbox{By the above derivation}\\
 			&= \sqrt{ N (1-\epsilon)q^k} \bigg[ \sum_{r \in S} \sqrt{p_r} \bigg] & \mbox{ } \\
 			&\leq \sqrt{ N (1-\epsilon)q^k} \sqrt{ \bigg[ \sum_{r \in S}  p_r \bigg] \cdot  |S| }  & \mbox{ By Cauchy–Schwarz inequality } \\
 			&\leq \sqrt{ N (1-\epsilon)q^k} \sqrt{ |S| } & \mbox{ Since $\sum_{r \in S} p_{r} \leq 1 $  } \\ 
 			&\leq \sqrt{ N (1-\epsilon)q^k} \sqrt{   N\cdot \frac{1}{(1-\epsilon)q^k}  }  = N, & \mbox{  Since we are in {\bf Case 2} } 
 		\end{align*}
 		
 		which is a contradiction. This implies that there exists some $\bar{r} \in S$  such that $C^p_{\bar{r}}(n_{\bar{r}}) \geq 1 - \frac{1}{(1-\epsilon)Nq^k}$.

 		Next, we argue that the existence of such a high cost resource $\bar{r}$ implies a condition on $p_{r}, n_r$ for any other resource $r \in S$. Specifically,  the fact that any agent at resource $\bar{r}$, suffering cost at least $1 - \frac{1}{(1-\epsilon)Nq^k}$, do not choose to move to any other resource $r \in S$ must because that will lead to a higher congestion cost after increasing the number of agents $n_{r}$ by $1$ at resource $r$. In other words, we must have for any $r \in S$ and $r \not = \bar{r}$: 
 		\begin{equation}\label{eq:xn-relation1}
 			1 - \frac{1}{(1-\epsilon)Nq^k} \leq C_{r}^p (n_{r} + 1) = \Big[ 1 + 2 \sum_{r' \in adj(r)} p_{r'} \Big]   - p_{r} \cdot \frac{1}{(n_{r} + 1)^2}.
 		\end{equation}   
 		Additionally, the above inequality holds trivially for $r = \bar{r}$ as well since $1 - \frac{1}{(1-\epsilon)Nq^k} \leq C_r^p(n_r) \leq C_r^p(n_r+1)$. Consequently, we have 
 		\begin{align*}
 			& \quad \texttt{SC}(\bvec{n}; x) \\
 			&= \sum_{r \in S} n_{r} \times \Bigg[ \Big[ 1 + 2 \sum_{r' \in adj(r)} p_{r'} \Big]   - p_{r} \cdot \frac{1}{(n _{r})^2} \Bigg]   & \mbox{by definition of \texttt{SC}} \\ 
 			&\geq \sum_{r \in S} n_{r} \times \Bigg[ 1 - \frac{1}{(1-\epsilon)Nq^k} + p_{r} \cdot \frac{1}{(n_{r} + 1)^2}  - p_{r} \cdot \frac{1}{(n _{r})^2} \Bigg]    & \mbox{by Inequality \eqref{eq:xn-relation1} } \\ 
 			&= N  - \frac{1}{(1-\epsilon)q^k} -  \sum_{r \in S} n_{r}    p_{r}  \times \big[ \frac{ 1}{(n_{r})^2 } - \frac{ 1}{(n_{r}+1)^2 }  \big]   & \mbox{since $\sum_{r \in S} n_{r} = N$} \\  
 			&= N - \frac{1}{(1-\epsilon)q^k} - \sum_{r \in S} p_{r}     \times \frac{2n_{r}+1}{n_{r}(n_{r}+1)^2} & \\  
 			&\geq N  - \frac{1}{(1-\epsilon)q^k} -  \sum_{r \in S} p_{r}     \times \frac{3}{4}  &  \mbox{plugging in $n_{r} = 1$ }\\ 
 			&\geq N - \frac{1}{(1-\epsilon)q^k}  - \frac{3}{4}    &  \mbox{ $\sum_{r\in S} p_{r} \leq 1$ }\\ 
 		\end{align*}

 		\subsubsection*{Situation 2: for any $r \in S$, $p_{r} \cdot \frac{1}{(n_{r}+1)^2} -   2 \Big[ \sum_{r' \in adj(r)} p_{r'}  \Big] \leq 0$.}
 		
 		According to the proof of Lemma~\ref{lem:SignalIndep},  the congestion at any resource $r$ is $r \in S$ is $C^p_r(n_r) = 1 +  2 \sum_{r' \in adj(r)} p_{r'}  -   p_{r} \cdot \frac{1}{(n _{r})^2}$. Therefore, the total social cost is computed as follows:
 		\begin{align*}
 			& \quad  \texttt{SC}(\bvec{n}; x)   \\
 			&= \sum_{r \in S} n_{r} \times \Bigg[ \Big[ 1 + 2 \sum_{r' \in adj(r)} p_{r'}  \Big]   - p_{r} \cdot \frac{1}{(n _{r})^2} \Bigg] + (N - \sum_{r \in S} n_{r}) & \mbox{by definition of \texttt{SC}} \\ 
 			&\geq \sum_{r \in S} n_{r} \times \Bigg[ 1 + p_{r} \cdot \frac{1}{(n_{r} + 1)^2} - p_{r} \cdot \frac{1}{(n _{r})^2} \Bigg]  + (N - \sum_{r \in S} n_{r}) & \mbox{by assumption of Situation 2} \\ 
 			&= N  - \sum_{r \in S} n_{r}   p_{r}  \times \big[ \frac{ 1}{(n_{r})^2 } - \frac{ 1}{(n_{r}+1)^2 }  \big]   & \mbox{} \\  
 			&= N - \sum_{r \in S} p_{r}     \times \frac{2n_{r}+1}{n_{r}(n_{r}+1)^2} & \\  
 			&\geq N  - \sum_{r \in S} p_{r}     \times \frac{3}{4}  &  \mbox{plugging in $n_{r} = 1$ }\\ 
 			&\geq N  - \frac{3}{4}    &  \mbox{ $\sum_{r\in S} p_{r} \leq 1$ }\\ 
 		\end{align*}

 		\section{Proof of Proposition~\ref{thm:private-hardness}.} \label{proof:private-hardness}

 		Our reduction is from the APX-hardness of the maximum independent set for 3-regular graphs. Specifically, there is an absolute constant $a \in (0,1)$ such that it is NP-hard to obtain a $a$-approximation for the maximum independent set problem for 3-regular graphs \cite{ chlebik2003approximation}.  
 		
 		Given any 3-regular graph $G = (V,E)$, we will interpret $V$ as the set of resources in our setup and construct the following instance of Optimization Problem  \eqref{incentive-deviate-SC}. For any agent $i \in [N]$, let $z^i_{r,r'} = - \frac{1}{4}$ if $(r,r') \in E$ is an edge in $E$ and let $z^i_{r,r'} = 0$ if $(r,r')$ is not an edge. We choose affine congestion function in the form $c_r(n_r) = n_r/N$ for all resources $r$. Instantiated in this particular instance, OP \eqref{incentive-deviate-SC} becomes the following problem 
 		\begin{eqnarray} 
 			&& \min_{\bv{n} \in P(A)} \quad \sum_{r \in [R]}  \frac{ (n_r)^2}{N} - \sum_{(r, r') \in E} n_r \bigg[ \frac{n_r}{N} - \frac{ (n_{r'}+1)}{N}  \bigg]\cdot \frac{1}{4}  \nonumber \\ 
 			&\Longleftrightarrow& \min_{\bv{n} \in P(N,R
 				)} \quad \sum_{r \in [R]}  \bigg( 1 - \sum_{(r, r')\in E}\frac{1}{4} \bigg) \bigg(\frac{n_r}{N}\bigg)^2 + \sum_{(r,r') \in E} \frac{n_r}{N}  \bigg(\frac{n_{r'}+1}{N}\bigg) \cdot \frac{1}{4} \nonumber \\ \label{incentive-deviate-SC-instance}
 			&\Longleftrightarrow&  \min_{\bv{n} \in P(A)} \quad \frac{1}{4} \sum_{r \in [R]} \bigg( \frac{n_r}{N} \bigg)^2 + \frac{1}{4} \sum_{(r,r')\in E} \frac{n_r}{N} \cdot \frac{n_{r'}}{N} + \frac{3}{4N}   
 		\end{eqnarray}
 		where (i) first equivalence is derived by basic calculation and then dividing both sides of the optimization problem by $N$; (ii) the second equivalence relation utilizes the fact the graph is 3-regular, thus each $r$ has exactly 3 neighbor nodes.   
 		
 		To prove that optimization problem  \eqref{incentive-deviate-SC-instance} is hard, our strategy is to first convert it to a continuous optimization problem and prove that it is NP-hard to obtain a constant approximation for the continuous optimization. We then show that the optimal objective of this continuous optimization problem is provably close to the optimal objective of OP \eqref{incentive-deviate-SC-instance}. This thus implies the hardness of solving OP \eqref{incentive-deviate-SC-instance} exactly.  
 		
 		Specifically, define new variable $\lambda_r = \frac{n_r}{N}$.  Since  $\sum_{r \in [R]} n_r = N$, we thus have vector $\lambda \in \Delta_{[R]}$. Note that by definition the entries of $\lambda$ are multiples of $1/N$.   Re-writing optimization problem  \eqref{incentive-deviate-SC-instance} with variable $\lambda$ and relax $\lambda$ to be a continuous variable in $\Delta_{[R]}$, we obtain the following continuous version of optimization problem  \eqref{incentive-deviate-SC-instance}: 
 		\begin{equation}
 			\label{incentive-deviate-SC-instance-cont}
 			\frac{3}{4N} + \min_{\lambda \in \Delta_{[R]}} \frac{1}{4} \bigg( \sum_{r \in [R]} (\lambda_r)^2 + \sum_{(r,r')\in E}  \lambda_r \lambda_{r'} \bigg)
 		\end{equation} 
 		It is well known that $\min_{\lambda \in \Delta_{[R]}} \big( \sum_{r \in [R]}     (\lambda_r)^2 + \sum_{(r,r')\in E} \lambda_r \lambda_{r'} \big) $ equals precisely $1/\texttt{MaxInd}$ where $\texttt{MaxInd}$ is the size of the maximum independent set of graph $G$ (see, e.g., \cite{ de2008complexity}). Since $G$ is a 3-regular graph in our construction and it is NP-hard to obtain a constant (multiplicative) approximation for $G$. This translate to the NP-hardness of multiplicative $\frac{1}{a}$ constant approximation for $\min_{\lambda \in \Delta_{[R]}} \,    \big( \sum_{r \in [R]} (\lambda_r)^2  +  \sum_{(r,r')\in E}  \lambda_r \lambda_{r'} \big) $.  This minimization problem has objective value at least $1/R$  since we know that its optimal objective is $1/\texttt{MaxInd} \geq 1/R$. NP-hardness of  multiplicative $\frac{1}{a}$  approximation implies the NP-hardness of additive $(\frac{1}{a} - 1)\cdot\frac{1}{R}$  additive approximation for Optimization Problem \eqref{incentive-deviate-SC-instance-cont}, which is the continuous version of the discrete optimization problem \eqref{incentive-deviate-SC-instance}. 
 		
 		Finally, it is not difficult to see that when $N = \Omega(R^3)$, the difference between the optimal objective of OP \eqref{incentive-deviate-SC-instance} and its continuous version \eqref{incentive-deviate-SC-instance-cont}  is at most $O(1/R^2)$, therefore any $poly(R,N)$-time algorithm exactly solving  OP \eqref{incentive-deviate-SC-instance} for $N = \Omega(R^3)$ will imply a $poly(R)$-time $O(1/R^2)$ additive approximation for OP \eqref{incentive-deviate-SC-instance-cont}, which however is NP-hard. This overall proves that OP \eqref{incentive-deviate-SC} is NP-hard when the game is symmetric and all resources have the same linearly increasing congestion functions $c_r(n) = \frac{n}{N}$.

 		%  We claim that there exists a constant $c>0$ such that it is NP-hard to obtain a $c$ Optimization Problem \eqref{incentive-deviate-SC-instance-cont} is N   Construct a 3-regular graph $G$ with all the locations $[L]$, such that for every location $l$, there are exactly 3 other locations $l_1, l_2, l_3 \neq l$ that has an edge with $l$. Let $z_{l, l'} = -\frac{1}{4}$ if $(l, l')$ is an edge in the 3-regular graph $G$ and $z_{l, l'} = 0$ otherwise. Let the coefficient $\alpha_l = -1$ for all $l$. The objective of the optimization problem becomes:
 		%     \begin{equation}
 			%         \sum_{l \in [L]} \left( -\frac{1}{4} \lambda_l^2 - \sum_{(l, l') \in G} \frac{1}{4} \lambda_l \lambda_{l'} \right) + \frac{3}{4N}
 			%     \end{equation}
 		%     We can ignore the constant term $\frac{3}{4N}$. The objective function is equivalent to $(-\frac{1}{4}) \lambda^T \cdot (I+A) \cdot \lambda$, where $\lambda \in \Delta_L$, $I$ is the $L \times L$ identity matrix, and $A$ is the adjacency matrix of the 3-regular graph $G$. This problem is NP-hard due to a inapproximability result for the maximum stable set problem by in \citet{ de2008complexity,haastad1999clique}. This concludes our proof.

 		\section{Example of the non-convexity of $SC^*(\C(\bv{p}))$.} \label{append:social-cost-non-convex-example}
 		
 		We construct a simple game with $N=3$ agents, $R=2$ resource, and two states of nature $\Theta = \{ \theta_1, \theta_2 \}$. Both resources are available for each agent, i.e., $A_1 = A_2 = \{1, 2\}$. Let two posterior distributions $\bv{p}_1, \bv{p}_2$ be deterministic at $\theta_1$, $\theta_2$ respectively. The cost functions at $\bv{p}_1$, $\bv{p}_2$ is described in the Table \ref{table:social-cost-non-convex-cost-functions}.
 		
 		\begin{table}[ht]
 			\centering 
 			\begin{tabular}{c|c c c|c c c}
 				\hline
 				\multicolumn{1}{c}{} & 
 				\multicolumn{3}{c}{Resource 1} & \multicolumn{3}{c}{Resource 2}\\
 				Posterior $\bv{p}$ & $c_1(1)$ & $c_1(2)$ & $c_1(3)$ & $c_2(1)$ & $c_2(2)$ & $c_2(3)$ \\ [0.5ex] 
 				\hline
 				$\bv{p}_1$ & 1 & 1 & 10 & 9 & 10 & 10 \\ 
 				$\bv{p}_2$ & 1 & 1 & 4 & 5 & 5 & 10 \\
 				\hline
 			\end{tabular}
 			\caption{Cost functions for 2 resources and posteriors $\bv{p}_1$, $\bv{p}_2$.} 
 			\label{table:social-cost-non-convex-cost-functions}
 		\end{table}
 		
 		It is easy to verify that for $\bv{p}_1, SC^*(\C(\bv{p}_1)) = 11$ when $2$ agents choose resource $1$ and $1$ agent chooses resource $2$; for $\bv{p}_2, SC^*(\C(\bv{p}_2)) = 12$ when all $3$ agents choose resource $1$.
 		
 		We then consider two convex combinations of $\bv{p}_1$ and $\bv{p}_2$. Specifically, let $\bv{p}_3 = 0.6 \bv{p}_1 + 0.4 \bv{p}_2$. We have $SC^*(\C(\bv{p}_3)) = 9.4$ when $2$ agents choose resource $1$ and $1$ agent chooses resource $2$. Since $SC^*(\C(0.6 \bv{p}_1 + 0.4 \bv{p}_2)) < 0.6 SC^*(\C(\bv{p}_1)) + 0.4 SC^*(\C(\bv{p}_2))$, $SC^*(\C(\bv{p}))$ is not a concave function.
 		
 		Let $\bv{p}_4 = 0.4\bv{p}_1 + 0.6\bv{p}_2$. We have $SC^*(\C(\bv{p}_4)) = 19.2$ when all $3$ agents choose resource $1$. Since $SC^*(\C(0.4 \bv{p}_1 + 0.6 \bv{p}_2)) > 0.4 SC^*(\C(\bv{p}_1)) + 0.6 SC^*(\C(\bv{p}_2))$, $SC^*(\C(\bv{p}))$ is not a convex function. Therefore, this game shows that $SC^*(\C(\bv{p}))$ is neither a convex function nor a concave function.

 \end{document}